\pdfoutput=1 
\pdfoutput=1 
\pdfoutput=1 
\pdfoutput=1 
\pdfoutput=1
\documentclass{IEEEtran}                                
\usepackage{graphics} 
\usepackage{epsfig} 
\usepackage{times} 
\usepackage{amsmath} 
\usepackage{amsthm}
\usepackage{amssymb}
\usepackage[colorlinks=false, urlcolor=blue, pdfborder={0 0 0}]{hyperref}
\usepackage{enumerate}
\usepackage{color}
\usepackage[linesnumbered,boxed,ruled,commentsnumbered]{algorithm2e}
\usepackage{subfigure}
\usepackage{comment}
\usepackage{cite}
\usepackage{flushend}
\title{Secure Your Intention:   On  Notions of Pre-Opacity in Discrete-Event Systems}

\author{Shuo~Yang,~\IEEEmembership{Student Member,~IEEE,}
	Xiang~Yin,~\IEEEmembership{Member,~IEEE}
	\thanks{This work was  supported by the National Natural Science Foundation of China (61803259, 61833012) and by Shanghai Jiao Tong University Scientific and Technological Innovation Funds.
	}
	\thanks{Shuo Yang and Xiang Yin  are with Department of Automation
		and Key Laboratory of System Control and Information Processing,
		Shanghai Jiao Tong University, Shanghai 200240, China.
		{\tt\small  \{xiang-yang,yinxiang\}@sjtu.edu.cn.}}
}

\newtheorem{mydef}{Definition}
\newtheorem{mythm}{Theorem}

\newtheorem{mypro}{Proposition}
\newtheorem{myexm}{Example}
\newtheorem{remark}{Remark}

\IncMargin{0.8em}
\begin{document}

	\maketitle
	
	\begin{abstract}
		This paper investigates an important information-flow security property called opacity in partially-observed discrete-event systems.   
		We consider the presence of a passive intruder (eavesdropper) that knows the dynamic model of the system and can use the generated information-flow to infer some ``secret" of the system. 
		A system is said to be opaque if it always holds the plausible deniability for its secret. 
		Existing notions of opacity only consider secret either as currently visiting some secret states or as having visited some secret states in the past.
		In this paper, we investigate information-flow security from a new angle by considering the secret of the system as the \emph{intention} to execute some particular behavior of importance in the future.
		To this end, we propose a new class of opacity called \emph{pre-opacity} that characterizes whether or not the intruder can predict the visit of secret states a certain number of steps ahead before the system actually does so. 
		Depending the prediction task of the intruder, we propose two specific kinds of pre-opacity called  \emph{$K$-step instant  pre-opacity} and \emph{$K$-step trajectory  pre-opacity} to specify this concept.
		For each notion of pre-opacity, we provide a necessary and sufficient condition as well as an effective verification algorithm.  
		The complexity for the verification of pre-opacity is exponential in the size of the system as we show that pre-opacity is inherently PSPACE-hard. Finally, we generalize our setting to the case where the secret intention of the system is modeled as executing a particular sequence of events rather than visiting a secret state. 
	\end{abstract}
	\begin{IEEEkeywords}
		Discrete-Event Systems,  Opacity, Prediction
	\end{IEEEkeywords}
	
	\section{Introduction}
	In the past decade, the notion of opacity has drawn a lot of attention in the Discrete-Event Systems (DES) literature as it provides a formal approach towards the verification and design of information-flow security for dynamics systems. 
	Roughly speaking, opacity is a confidentiality property that captures whether or not the information-flow generated by a dynamic system can reveal some ``secret behavior"  to an outside observer (intruder) that is potentially malicious.
	In other words, an opaque system should always maintain the plausible deniability for its secret behavior during its execution. 
	In the context of DES, opacity has been extensively studied for different system models including
	finite-state automata \cite{saboori2011verificationa,saboori2013verification,yin2017new,mohajerani2019transforming}, labeled transition systems \cite{bryans2008opacity,kuizecdc17} and Petri nets \cite{bryans2005modelling,tong2017verification,Tong17OpacityPetriNets,cong2018line}. 
	More recently, opacity has been extended to continuous dynamic systems with possibly infinite state spaces and time-driven dynamics \cite{ramasubramanian2020notions,an2020opacity,yin2020approximate}.
	Many enforcement techniques have also been proposed when the original system is not opaque; see, e.g.,  
	\cite{takai2008formula,dubreil2010supervisory,cassez2012synthesis,darondeau2014enforcing,zhang2015max,ji2018enforcement,wu2018synthesis,tong2018current,behinaein2019optimal,mohajerani2019compositional}.
	Opacity has also been applied to certify/enforce security in many real-world systems including mobile robots \cite{saboori2011coverage}, location-based services \cite{wu2014ensuring}, battery management systems \cite{lin2020information} 
	and web services \cite{bourouis2017verification}. 
	The reader is referred to the survey papers \cite{jacob2015opacity,lafortune2018history} for more details on opacity and its recent developments.
	
	In order to capture different security requirements, different notions of opacity have been proposed in the literature.  
	For example, in language-based opacity \cite{lin2011opacity}, the secret is formulated as the executions of some particular secret strings. As shown in \cite{wu2013comparative}, this formulation is equivalent to the notion of current-state opacity, where the secret is formulated as a set of secret states 
	and a system is current-state opaque if the intruder cannot determine for sure that the system is currently at a secret state. 
	In some situations, the system may want to hide its initial location or its location at some specific previous instant; such requirements can be captured by initial-state opacity \cite{saboori2013verification} and $K$/infinite-step opacity \cite{saboori2011verificationa,saboori2011verificationb,falcone2015enforcement,yin2017new}, respectively.  
	More recently, quantitative notions of opacity have been proposed for stochastic DES in order to measure the secret leakage of the system; see, e.g.,  \cite{saboori2014current,berard2015probabilistic,keroglou2016probabilistic,chen2017quantification,wu2018privacy,yin2019infinite}.

	As we can see from the above discussion,  ``secret" in opacity analysis is actually a generic concept. 
	Based on what kind of information the user would like to hide, or equivalently, how the intruder can utilize information to infer the secret of the system,  existing notions of opacity in the literature as reviewed above can generally be  divided into the following  two categories: 
	\begin{itemize}
		\item
		Opacity for Current Information: the intruder wants to determine the current behavior of the system based on the current observation. In other words, the user does not want the outsider to know for sure that  it \emph{is currently doing} something secret. This category includes, e.g.,  current-state opacity  and language-based opacity.    
		\item
		Opacity for Delayed Information: the intruder wants to determine the  previous secret behavior of the system at some instant based on the current observation.  In other words,  the user does not want the outsider to know for sure that it \emph{has done} something secret at some previous instant. 
		This category includes, e.g., initial-state opacity, $K$-step opacity  and infinite-step opacity.  
		Note that delayed information is involved here as the intruder does not need to specify the visit of a secret state immediately; it can use future information to improve its knowledge about the previous instants.   
	\end{itemize}
	There are also some works that combine these two types of opacity together, e.g., by combing current-state opacity and initial-state opacity, one can define the notion of initial-final-state opacity  \cite{wu2013comparative}.

	In some applications, however, the ``secret" one wants to hide can be its \emph{intention} to do something of particular importance in the future. As a simple motivating example, let us consider a single robot moving in a region whose mobility is described by a DES shown in Figure~\ref{fig:G}, where each state represents a location and each transition represents an action. Some actions are assumed to be observable by outsider; $E_{o}=\{o_1,o_2,o_3\}$ are observable actions. The robot may choose to attack state $9$ by reaching it. 
	However, it does not want to reveal its intention to  attack  state $9$ too early; otherwise, e.g., some defense  strategy can be implemented in advance.  Clearly, the shortest path to reach state $9$ is
	$0\xrightarrow{o_1}3\xrightarrow{o_2}6\xrightarrow{o_3}9$. However, by doing so, the outsider will know the robot's intention of attack two steps ahead just by observing the first action $o_1$. On the other hand, the robot can choose to attack state $9$ via path  $0\xrightarrow{u_2}2\xrightarrow{o_2}5\xrightarrow{o_1}8\xrightarrow{o_2}11\xrightarrow{o_3}9$, 
	which  is longer but allows the robot to hide its intention of visiting state $9$ until it actually reaches it. This is because this path has the same observation of $0\xrightarrow{u_1}1\xrightarrow{o_2}4\xrightarrow{o_1}7\xrightarrow{o_2}10$ whose continuation may not necessarily be secret. 
	Existing notions of opacity cannot capture this scenario as this problem  essentially requires another type of opacity for \emph{future information}: the user does not want the outsider to know too early for sure that it \emph{will do} something secret at some future instant. 
	
	\begin{figure}
		\center
		\includegraphics[width=0.3\textwidth]{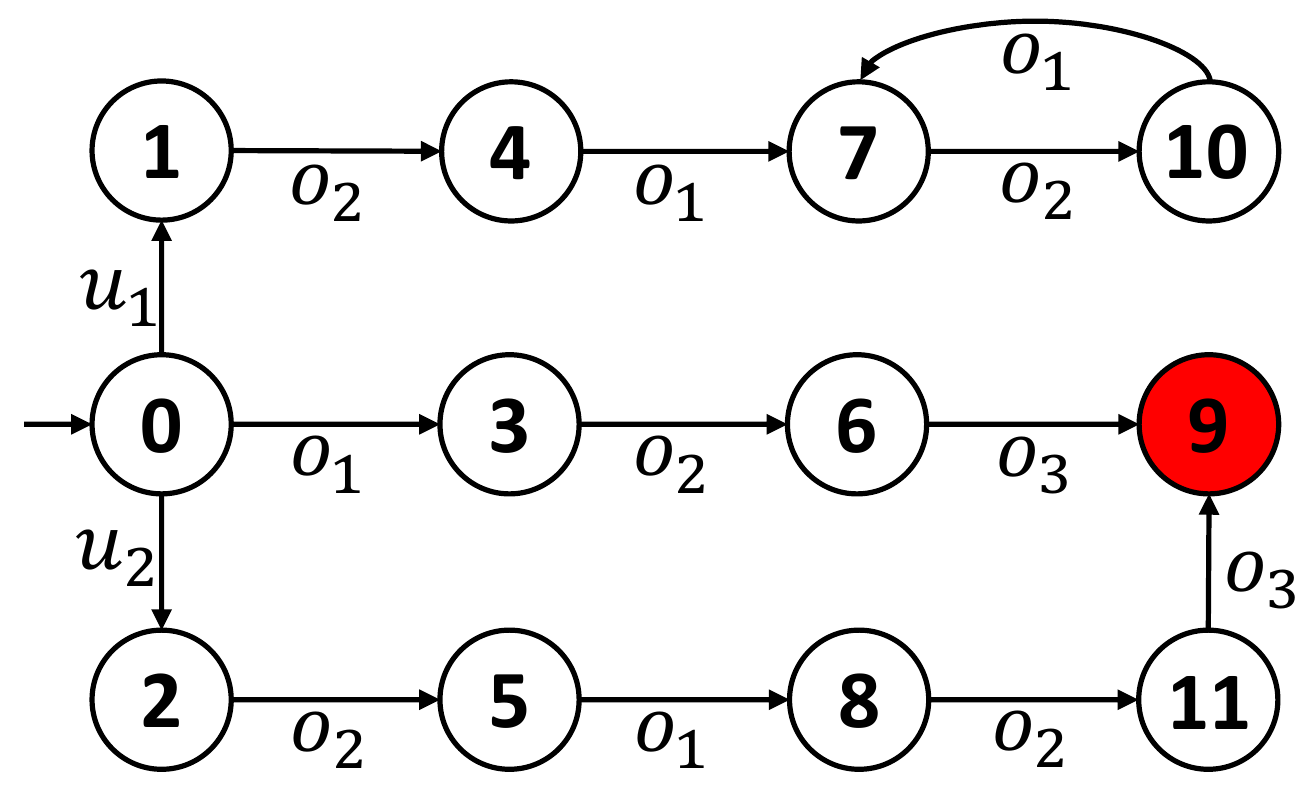}
		\caption{A motivating example with $E_o=\{a,b,c\}$ and $E_{uo}=\{u\}$. State 6 is the target (secret) state. }\label{fig:G}
	\end{figure}
	
	In this paper, we investigate opacity from a new angle by considering the system's \emph{intention} of executing some particular behavior as the secret.  Then we propose a new type of opacity, called \emph{pre-opacity},  to characterize whether or not the secret intention of the system can be revealed. We follow the standard setting of opacity by considering a passive intruder modeled as an eavesdropper  that  knows the model of the system. Then we propose two notions of pre-opacity called   \emph{$K$-step instant  pre-opacity} and \emph{$K$-step trajectory  pre-opacity}; the former requires that the intruder cannot determine $K$-step ahead for sure that the system will  be at secret states \emph{for some specific instant}, while the latter requires that the intruder cannot determine $K$-step ahead for sure that  the system will visit  secret states in the future without the need of specifying the instant of being secret.  
	Properties of these two notions of pre-opacity are investigated and we show that instant  pre-opacity is strictly weaker than trajectory  pre-opacity. 
	Furthermore, for each pre-opacity, we provide necessary and sufficient condition as well as effective verification algorithm. 
	We show that both properties are PSPACE-hard; hence the exponential verification complexity is unavoidable. 
	Also, we discuss  the case where ``secrets" are modeled as a \emph{sequence pattern}  rather secret states.  
	
	In the systems theory, there are three fundamental types of estimations problems: filtering, smoothing and prediction. 
	Essentially, current-state opacity can be viewed as the plausible deniability for secret under filtering and infinite/$K$-step opacity can be viewed as the plausible deniability for secret under smoothing. Analogously, the proposed notion of pre-opacity can also be interpreted as the plausible deniability for secret under prediction. 
	Therefore, our new notion also generalizes the framework of opacity from the systems theory point of view.

	The proposed notion of pre-opacity, in particular, trajectory pre-opacity, is closely related to the notion of fault predictability (or prognosability) in the literature; see, e.g., 
	\cite{jeron2008pred,genc2009predictability,kumar2010decentralized,takai2015robust,chen2015stochastic,yin2016decentralized}.   
	However, predictability requires that  any fault can be predicted before its occurrence, but our notion of pre-opacity requires that any secret cannot be predicted before it  actually happens. Furthermore, our notion of instant pre-opacity is much more different since it requires to determine the precise instant of being secret, which is not required in predictability analysis.  
	Also, in fault prediction problems, once the system becomes faulty, it is faulty forever. 
	However, in pre-opacity analysis, the system's behavior can become secret/non-secrete intermittently in the sense that, even when the intruder fails to predict the first secret behavior, it may still has chance to predict some future secret so that the security of the system can still be  threatened.  
	Therefore, although predictability is conceptually related to our notion of pre-opacity, these two properties are technically very different.
	
	The rest of the paper is organized as follows. 
	In Section~\ref{sec:2}, we describe the system model and review the existing notions of opacity. 
	Section~\ref{sec:3} introduces the two new notions of pre-opacity and discusses their properties. 
	In Section~\ref{sec:4}, we provide effective algorithms for the verification of notions of pre-opacity. 
	The proposed pre-opacity is further generalized to the case of sequence pattern in Section~\ref{sec:5}. 
	Finally, we conclude this paper by Section~\ref{sec:6}.
	
	\section{Preliminaries}\label{sec:2}
	
	\subsection{System model}
	Let $E$ be a finite set of events.  
	A \emph{string} is a finite sequence of events and we denote by $E^*$ the set of all strings over $E$ including the empty string $\epsilon$. 
	For any string $s\in E^*$, we denote by $|s|$ the length of $s$ with $|\epsilon|=0$.
	A language $L\subseteq E^*$ is a set of strings, and $\bar{L}$ denotes the prefix-closure of $L$, i.e., $\bar{L}=\{u\in E^*: \exists v\in E^*\text{ s.t. }  uv\in L\}$.
	
	We consider a discrete-event system  modeled by a deterministic finite-state automaton  (DFA)
	\[
	G=(X,E,f,X_0),
	\]
	where $X$ is the finite set of states, $E$ is the finite set of events, $f:X\times E \rightarrow X$ is the partial deterministic transition function such that $f(x,\sigma)=x'$ means that there exists a transition from $x$ to $x'$ with event label $\sigma$, and $X_0\subseteq X$ is the set of initial states. 
	The transition function $f$ is also extended to  $f:X\times E^*\rightarrow X$ recursively by: 
	for any $x\in X,s\in E^*,\sigma\in E$, we have $f(x,s\sigma)=f(f(x,s),\sigma)$ with $f(x,\epsilon)=x$.
	
	The language generated by $G$ from state $x\in X$ is defined by 
	$\mathcal{L}(G,x)=\{s\in E^*:f(x,s)!\}$, where ``$!$" means ``is defined''.
	Also, we define $\mathcal{L}(G,Q):=\bigcup_{x\in Q}\mathcal{L}(G,x)$ as the language generated from a set of states $Q\subseteq X$.
	Therefore, the language generated by $G$ is $\mathcal{L}(G):=\mathcal{L}(G,X_0)$. 
	For the sake of simplicity, hereafter, we assume that the system $G$ is live, i.e., 
	for any $x\in X$, there exists $\sigma\in \Sigma$ such that $f(x,\sigma)!$. 
	In some situations, a DFA is also equipped with a set of \emph{marked states} $X_m\subseteq X$ and we write a DFA with marked states as $G=(X,E,f,X_0,X_m)$. 
	Then the marked language of $G$ is 	$\mathcal{L}_m(G)=\{s\in E^*:\exists x_0 \in X_0 \text{ s.t. }f(x,s)\in X_m\}$. 
	
	\subsection{Intruder Model and Opacity}
	Following the standard setting of opacity, we assume that the intruder is modeled as a \emph{passive observer} (eavesdropper), which has the full knowledge of the system's structure. By ``passive", we mean  that the intruder can only  observe some behavior generated by the system, but it cannot actively affect the behavior of the system. 
	Formally, we assume that the event set $E$ is partitioned as:
	\[
	E=E_o\dot{\cup}E_{uo},
	\]
	where $E_o$ and $E_{uo}$ are the set of observable events and the set of unobservable events, respectively. 
	The natural projection from $E$ to $E_o$ is a mapping $P:E^*\rightarrow E^*_o$ defined recursively by:
	\begin{align}
	P(\epsilon)=\epsilon \text{ \ and \ }
	P(s\sigma)=
	\left\{
	\begin{array}{l l}
	P(s)\sigma\quad
	&\text{if } \sigma \in  E_o   \\
	P(s)\quad
	&\text{if } \sigma \notin  E_{uo}
	\end{array}
	\right.
	\end{align}
	The natural projection is also extended to $P:2^{E^*}\rightarrow 2^{E^*_o}$, i.e., $P(L)=\{t\in E^*_o: \exists s\in L \text{ s.t. }P(s)=t\}$ for any $L\subseteq E^*$.
	
	When string $s\in \mathcal{L}(G)$ is generated by the system, the intruder observes $P(s)$
	and it can use this observation together with the dynamic model of the system to infer which state the system could be in at some specific instant. 
	In opacity analysis, it is assumed that the system has a set of secret states, denoted by $X_S\subseteq X$.  
	Roughly speaking, a system is said to be opaque if the intruder can never determine for sure that the system is/was at a secret state based on its observation. 
	Here, we review the notion  of  $K$-step opacity which can be used to define current-state opacity and infinite-step opacity.

	\begin{mydef}($K$-Step Opacity)\label{def:opa} 
		Given system $G$,  set of observable events $E_o$, set of secret states $X_S$, and non-negative integer $K\in \mathbb{N}$, system $G$ is said to be \emph{$K$-step opaque}  (w.r.t.\ $E_o$ and $X_S$) if 
		\begin{align}
		&(\forall x_0\in X_0,\forall st\in \mathcal{L}(G,x_0): f(x_0,s)\in X_S\wedge |P(t)|\leq K) \nonumber \\
		&(\exists x_0'\in X_0)(\exists s't'\in \mathcal{L}(G,x'_0)) \text{ s.t.}\nonumber\\ 
		&[P(s)=P(s')]  \wedge  [ P(t)=P(t')] \wedge  [ f(x_0',s')\notin X_S].
		\end{align} 
		Furthermore,   system $G$ is said to be
		\begin{itemize}
			\item 
			\emph{current-state opaque} if it is $0$-step opaque;  
			\item 
			\emph{infinite-step opaque}  if it is $K$-step opaque for any $K\geq 0$.
		\end{itemize}
	\end{mydef}
	
	Intuitively, $K$-step opacity says that, whenever the system visits a secret state, 
	it should be able to keep this secret unrevealed within the next $K$ steps. In other words, the intruder should never be able to determine that the system was at a state secret for any instant in the past $K$ steps. 
	Note that current-state opacity can be viewed as a special case of $K$-step opacity ($K=0$) as it essentially requires that the intruder cannot determine for sure that the system is currently at a secret state. 
	To verify current-state opacity, one can construct the current-state estimator (or the observer) and check whether or not there exists a reachable  estimator state that only contains secret states. The verification of $K$-step opacity and infinite-step opacity are more involved as they require the computation of  delayed state estimate, which can be done by constructing the two-way observer\cite{yin2017new}.

	\begin{figure}
		\centering
		\subfigure[$G_1$]{\label{fig:G1}
			\includegraphics[width=0.23\textwidth]{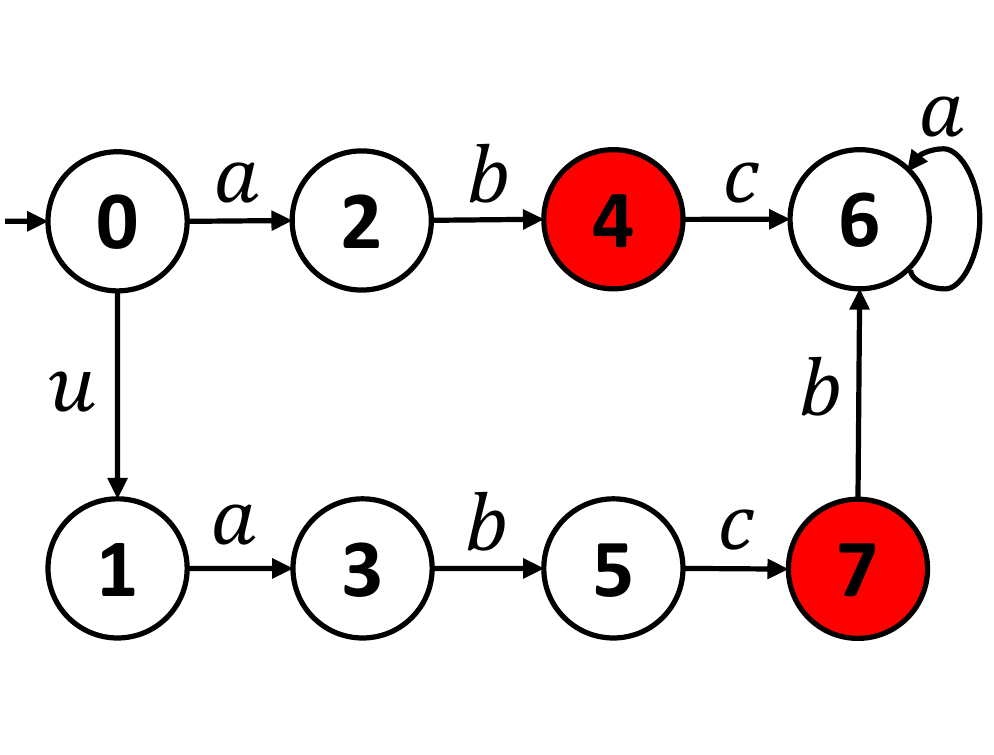}}
		\subfigure[$G_2$]{\label{fig:G2}
			\includegraphics[width=0.23\textwidth]{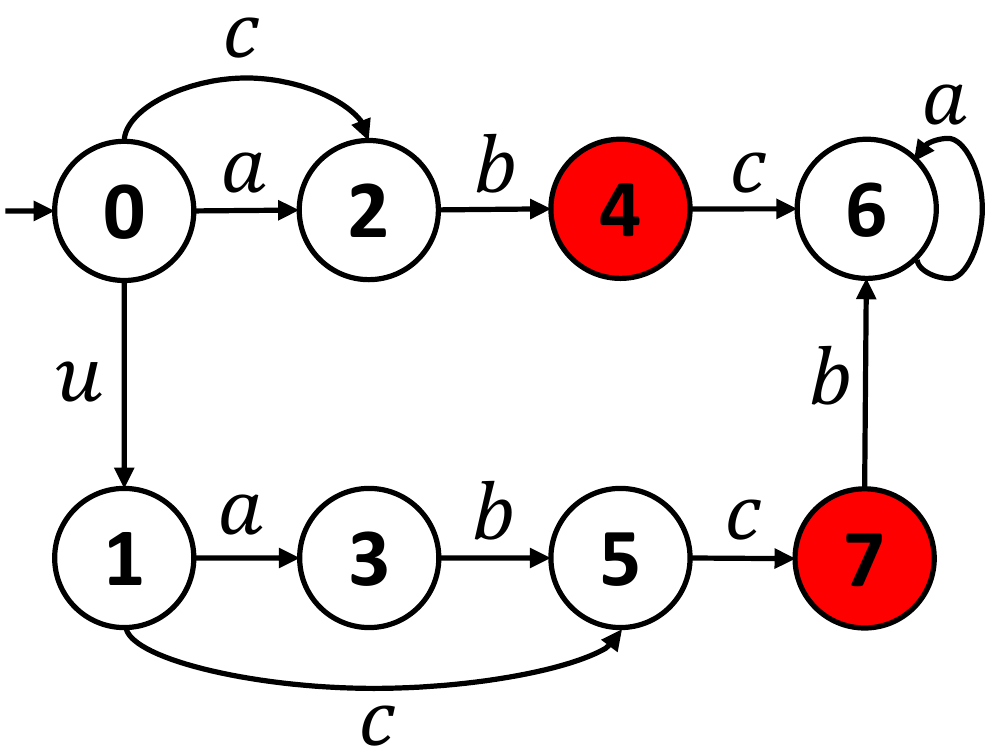}}
		\caption{For both systems, we have $X_0=\{0\},X_S=\{4,7\}$ and $E_o=\{a,b,c\}$.}. \label{fig:exm1}
	\end{figure} 
	
	\begin{myexm} 
		Let us consider system $G_1$ shown in Figure~\ref{fig:G1}, where  $X_S=\{4,7\}$ and $E_o=\{a,b,c\}$. Clearly, this system is current-state opaque. For example, by observing $ab$, the intruder cannot determine whether the system is at secret state $4$ or at non-secret state $5$ since $P(ab)=P(uab)=ab$. Similarly, when secret state $7$ is reached via $uabc$, the intruder still  cannot distinguish this state from non-secret state $6$. 
		On the other hand, this system is not $1$-step opaque. This is because, by observing $abcb$, the intruder can determine for sure that the system was at secret state $7$ one step ago. Therefore, $G_1$ is also not infinite-step opaque.
	\end{myexm}
	\section{Notions of Pre-Opacity}\label{sec:3}
	In this section, we first provide the definitions of $K$-step instant  pre-opacity and $K$-step trajectory  pre-opacity for DES. Then we discuss  properties of the proposed notions of pre-opacity. 
	
	\subsection{Definitions of $K$-step Instant/Trajectory Pre-Opacity}
	In the existing notions of opacity, secret is either  characterized by 
	whether  the system is doing something secret (current-state opacity) or characterized by whether the system has done something secret ($K$-step and infinite-step opacity). 
	These settings essentially assume that the system is operating against an intruder whose functionality is a current-state estimator or a delayed state-estimator.
	
	However, in some applications, what the system wants to hide might be its \emph{intention}, i.e.,   maintain the plausible deniability for its    willing to do something  secret in the future. In this setting, the system is essentially operating against an intruder that can be interpreted as a \emph{predictor}. 
	More specifically, the user may require that the intruder should never be able to determine its intention of visiting a secret state too early, which is characterized by a parameter $K$. 
	To this end, we first propose the notion of \emph{$K$-step instant pre-opacity} as follows; the reason why we use terminology ``instant" here will be clear soon. 
	
	\begin{mydef}($K$-Step Instant Pre-Opacity)\label{def:ins}   
		Given system $G$,  set of observable events $E_o$, set of secret states $X_S$, and non-negative integer $K\in \mathbb{N}$, system $G$ is said to be  $K$-step instant pre-opaque (w.r.t.\ $E_o$ and $X_S$) if 
		\begin{align}
		&(\forall x_0\in X_0,\forall s\in \mathcal{L}_o(G,x_0))(\forall n \geq K ) \nonumber \\
		&(\exists x_0'\in X_0,\exists s'\in \mathcal{L}_o(G,x'_0) , \exists t\in \mathcal{L}(G,f(x_0',s'))\text{ s.t. } \nonumber\\
		&[P(s)=P(s')]  \wedge  [|t|=n]  \wedge [f(x_0',s't)\notin X_S]
		\end{align}
		where 
		\[
		\mathcal{L}_o(G,x):=(\mathcal{L}(G,x) \cap E^*E_o) \cup \{\epsilon\}
		\]
		is the set of strings generated from $x$ that end up with observable events including the empty string. 
	\end{mydef}
	
	Intuitively,   $K$-step instant pre-opacity requires that, for any string $s$ generated from some initial state $x_0$ and any future instant $n\ge K$, there exists another observation-equivalent string $s'$ generated from some initial state $x_0'$ such that $s'$ can reach a non-secret state in exact $n$ steps. 
	In other words, the intruder can never determine more than $K$ steps ahead, based on its current observation, that the system will visit a secret state at some future instant.  
	Therefore, $K$ can be viewed as a parameter that determines \emph{how early} the user does not want to reveal its intention. For instance, if $K=2$, then the user may allow the intruder to determine just one step ahead that it will visit a secret system. 
	We use the following example to illustrate this notion.

	\begin{myexm}\label{exm:1} 
		First, let us consider again system $G_1$ in Figure~\ref{fig:G1}. 
		One can easily check that this system is $1$-step instant pre-opaque. 
		For example, for string $a\in \mathcal{L}_o(G)$, the intruder cannot predict for sure that the system will be at a secret in one step since there exist  another string $ua$ and its one-step extension $b$ such that $P(ua)=P(a)$ but $f(0,uab)=5\notin X_S$. 
		Similarly, the  intruder also cannot predict for sure that the system will be at a secret after  $2$ steps. 
		For example, when observing $\epsilon$, the system may reach non-secret state $3$ in two steps, which protects the possible secret intention of executing $ab$; 
		when observing $a$, the system may reach  non-secret state $6$ in two step, which protects the possible secret intention of executing $uabc$. 
		
		However, for system $G_2$ in Figure~\ref{fig:G2}, where $X_S=\{4,7\}$ and $E_o=\{a,b,c\}$, one can check that this system is not $1$-step instant pre-opaque. 
		This is because, by observing $c$, the intruder can determine for sure that the system is either at state $2$ or at state $5$. However, from either state $2$ or $5$, the system will reach a secret state in the next step. Therefore, its intention of visiting secret states will be revealed one step before it actually happens. 
	\end{myexm}

	\begin{remark} 
		In Definition~\ref{def:ins},  ``step" is counted by the number of occurrences of \emph{actual events} rather than the occurrences of observable events, i.e., we consider $|t|=n$ rather than $|P(t)|=n$. We believe this setting is more natural for predicting future instants.  Furthermore, we consider string $s$ in $\mathcal{L}_o(G,x_0)$ rather than $\mathcal{L}(G,x_0)$. 
		This implicitly assumes that the intruder will make a prediction immediately after observing a new observable event. 
		Hereafter, we will introduce the main developments based on this setting. 
	\end{remark}
	
	Note that $K$-step instant pre-opacity requires that the intruder cannot predict $K$-step ahead that the system will visit a secret state at some \emph{specific instant}.  This is also why we call it ``instant" pre-opacity. 
	However, in some situations, the intruder may just want to know whether or not the system will visit a secret state in the future without the need of telling the specific instant. 
	For instance, for $G_2$ in Figure~\ref{fig:G1}, after observing $a$, although the intruder cannot determine for sure the specific instant when the secret state will be reached (the system will visit a secret state in one step or in two steps), it can still tell that the system will visit a secret state within the next two steps and at least one step before the occurrence of the first secret state. 
	To capture this scenario, we propose the notion of $K$-step trajectory pre-opacity. 
	
	\begin{mydef}($K$-Step Trajectory Pre-Opacity)   \label{def:traj} 
		Given system $G$,  set of observable events $E_o$, set of secret states $X_S$, and non-negative integer $K\in \mathbb{N}$, system $G$ is said to be $K$-step trajectory pre-opaque (w.r.t.\ $E_o$ and $X_S$) if 
		\begin{align}
		&(\forall x_0\in X_0,\forall s\in \mathcal{L}_o(G,x_0))(\forall n\ge K)  \nonumber \\
		&(\exists x_0'\in X_0,\exists s'\in \mathcal{L}_o(G,x'_0) , \exists t_1t_2\in \mathcal{L}(G,f(x_0',s'))\text{ s.t. } \nonumber\\
		&[P(s)=P(s')]  \wedge  [|t_1|=K]\wedge  [|t_1t_2|=n]  \wedge \nonumber\\
		&\quad[\forall w\!\in\! \overline{\{t_2\}}:f(x_0',s't_1w)\!\notin\! X_S]\nonumber
		\end{align}
	\end{mydef}
	
	Intuitively, $K$-step trajectory pre-opaque says that the intruder will never be able to determine $K$-step ahead for sure that the system will visit a secret state. More specifically, if a system is not $K$-step trajectory pre-opaque, then according to Definition~\ref{def:traj}, it means that there exist  a string $s$ and an integer $n\geq K$ such that any observation equivalent string $s'$ must pass a secret state between the next $K$th instant and the next $n$th instant in the future. In other words, the intruder can determine the system's intention of visiting a secret state more than $K$-step ahead. We use the following example to illustrate this notion.
	
	\begin{myexm} 
		Let us consider again $G_1$ shown in Figure~\ref{fig:G1} and we have shown in Example~\ref{exm:1} that this system is $1$-step instant pre-opaque. 
		However, it is not $1$-step trajectory pre-opaque. 
		For example, let us consider $\epsilon\in \mathcal{L}_o(G)$ and $n=4\geq 1=K$. 
		Note that $\epsilon$ itself is the only observation equivalent string in $\mathcal{L}_o(G)$. 
		However, any $4$-step extension of $\epsilon$, either $abca$ or $uabc$,  will necessarily pass a secret state between the first instant and the forth instant.
		On the other hand, this system is $3$-step trajectory pre-opaque.  
		This is because the only instant to predict the visit of a secret state $3$-step ahead is when observing $\epsilon$.  
		However, with this observation, it is possible that the system will be at state $6$, from which no secret state will be visited,  after three steps. Therefore, the intruder can never determine $3$-step ahead for sure that a secret state will be visited.
	\end{myexm} 
	
	\begin{remark} 
		Similar to the interpretations of current-state opacity and $K$-step opacity, where the system is operating against the current-state estimator and delay-state estimator, respectively, here one can image that the system is operating against an intruder working as a \emph{predictor} (for its secret intention). 
		Roughly speaking, both   $K$-step instant pre-opacity and   $K$-step trajectory pre-opacity require that the intruder can never predict its secret  $K$-steps ahead. However, the specific prediction tasks of the ``virtual predictor"   in these two notions are different: in instant pre-opacity, the predictor also needs to identify the precise future instant at which the system will be at a secret state, while in trajectory pre-opacity, the predictor just needs to identify the inevitability of passing through a secret state without specifying the visiting instant.
	\end{remark}
	
	\subsection{Properties of Pre-Opacity}
	\begin{figure*}
		\centering
		\includegraphics[width=0.6 \textwidth]{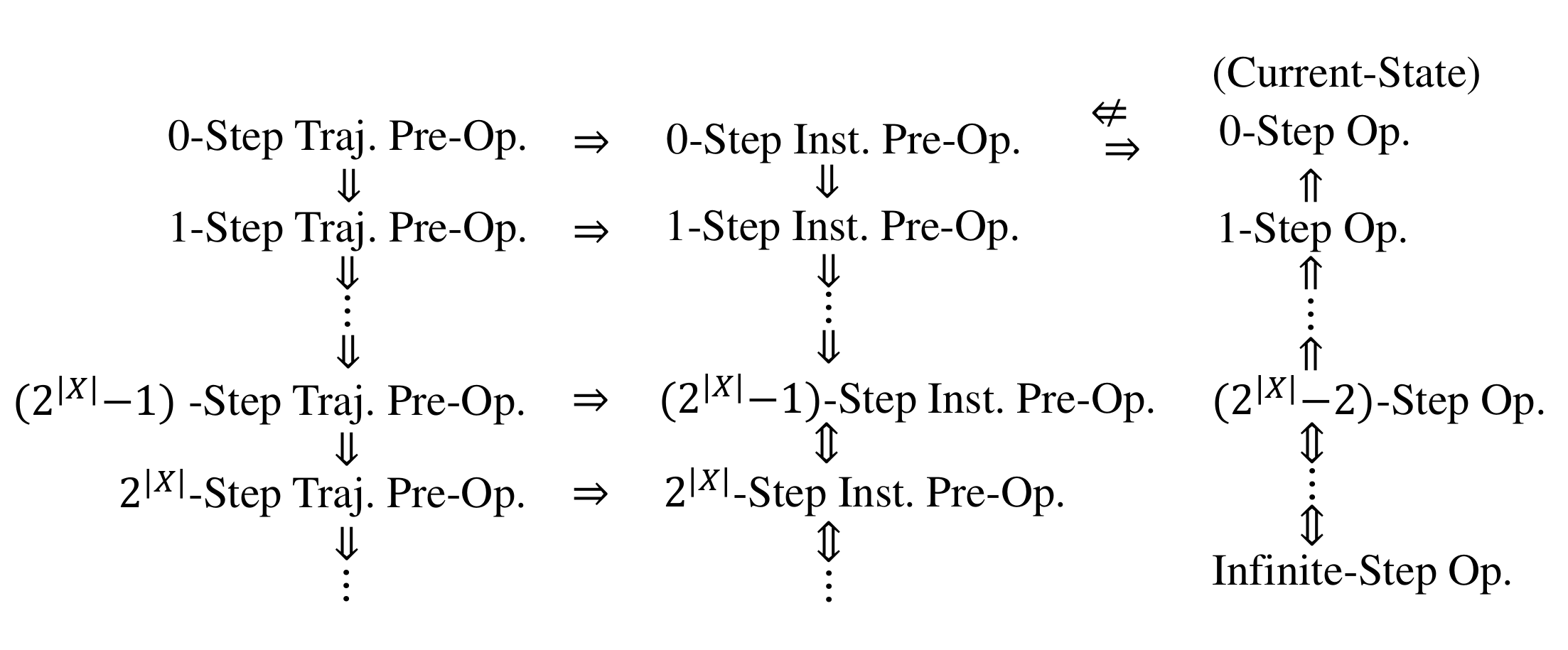} 
		\caption{Relationships among different notions of opacity and pre-opacity.} \label{fig:relation}
	\end{figure*} 
	Now, we discuss  properties of the proposed notions of pre-opacity and their relationships with other notions of opacity in the literature. 
	First, we show that, for any $K$,  $K$-step instant pre-opacity is  weaker than   $K$-step trajectory pre-opacity. 
	
	\begin{mypro} 
		If $G$ is $K$-step trajectory pre-opaque, then it is $K$-step instant  pre-opaque. 
	\end{mypro}
	\begin{proof}
		This result follows directly from the definitions. 
		If the system is   $K$-step trajectory pre-opaque, then by setting $t$ in 
		Definition~\ref{def:ins} as $t_1t_2$ in Definition~\ref{def:traj}, we know that the system is $K$-step instant pre-opacity.
	\end{proof}
	
	The intuition of the above result can also be interpreted as follows. For the case of instant pre-opacity, the prediction task of intruder is more challenging than that for the case of trajectory opacity due to the need of determining the specific secret instant. Therefore, from the system's point of view, the underlying security property becomes weaker.  
	
	Also, by definitions, we note that $K$-step instant pre-opacity becomes weaker when $K$ increases, 
	i.e., $K$-step instant pre-opacity always implies $(K+1)$-step instant pre-opacity. 
	However, the following  result  shows that there is an upper bound for $K$ in instant pre-opacity, i.e., pre-opacity will not keep getting strictly weaker when $K$ increases.  
	
	To present our result, we introduce two necessary concepts. First, for each state $x\in X$, the set of states that can be reached from $x$ in exactly $K$ steps  is define by
	\begin{align} 
	R_K(x)\!=\!\{ x'\!\in\! X:\exists s\!\in\! \mathcal{L}(G,x)\text{ s.t. }f(x,s)\!=\!x'\wedge |s|\!=\!K  \}.  
    \end{align}
    For a set of states $q\subseteq X$, we  also define $R_K(q):=\bigcup_{x\in q}R_K(x)$ as the set of states that can be reached from set $q$ in exactly $K$ steps.
		
	Also, let $\alpha\in P(\mathcal{L}(G))$ be an observed string. 
	Then the current-state estimate upon the occurrence of $\alpha$ without the unobservable tail is defined by
	\begin{align}
	\hat{\mathcal{E}}(\alpha)\!=\!\{ f(x_0,s)\!\in\! X: \exists x_0\!\in\! X_0, s\!\in\! \mathcal{L}_o(G,x_0)\text{ s.t. } P(s)\!=\!\alpha     \}.
	\end{align}
	
	Then we have the following the theorem showing the upper bound of $K$ in instant pre-opacity.
	\begin{mythm}\label{thm:bound}
		For any $K'>K\ge 2^{|X|}-1$, system $G$ is $K'$-step instant pre-opaque, if and only if, $G$ is $K$-step instant pre-opaque.
	\end{mythm}
	
	\begin{proof}
		It is trivial that $K$-step instant pre-opacity implies $K'$-step instant pre-opacity.
		Hereafter, we show that $K'$-step instant pre-opacity also implies $K$-step instant pre-opacity. 
		Without loss of generality, we assume that $K'=K+1$ as the
		argument can be applied inductively.
		
		Now we assume, for the sake of contradiction, that $G$ is not $K$-step instant pre-opaque but $G$ is $(K+1)$-step instant pre-opaque, where $K\ge 2^{|X|}-1$. 
		This implies that there exists an initial state $x_0\in X_0$ and a string $s\in \mathcal{L}_o(G,x_0)$ such that
		\begin{align}
		&(\forall x_0'\in X_0)(\forall s'\in \mathcal{L}_o(G,x'_0), s't\in \mathcal{L}(G,x'_0)) \nonumber\\
		&[P(s)=P(s')  \wedge  |t|=K  ] \Rightarrow [ f(x_0',s't)\in X_S].\nonumber
		\end{align}
		Equivalently, we have $R_K(\hat{\mathcal{E}}(P(s)))\subseteq X_S$. 
		Since for any $i\in \mathbb{N}$, $R_{i}(\hat{\mathcal{E}}(P(s)))$ is non-empty and it has at most $|X|$ elements, there are only $(2^{|X|}-1)$ choices for $R_{i}(\hat{\mathcal{E}}(P(s)))$. 
		Moreover, since the cardinality of multi-set $\{R_{j}(\hat{\mathcal{E}}(P(s))):j=0,1,\cdots,K\}$ is $K+1\ge 2^{|X|}>2^{|X|}-1$, 
		we know that there exist two integers $0\le m<n\le K$, such that $R_{m}(\hat{\mathcal{E}}(P(s)))$ = $R_{n}(\hat{\mathcal{E}}(P(s)))$. 
		Then we know that
		\begin{align}
			&\!\!\!\!\!\!R_{K+n-m}(\hat{\mathcal{E}}(P(s)))=R_{K-m}(R_{n}(\hat{\mathcal{E}}(P(s))))\nonumber\\
			&=R_{K-m}(R_{m}(\hat{\mathcal{E}}(P(s))))=R_{K}(\hat{\mathcal{E}}(P(s)))\subseteq X_S
		\end{align}
		i.e., for initial state $x_0\in X_0$ and string $s\in \mathcal{L}_o(G,x_0)$, we also have that
		\begin{align}
		&(\forall x_0'\in X_0)(\forall s'\in \mathcal{L}_o(G,x'_0), s't\in \mathcal{L}(G,x'_0)) \nonumber\\
		&[P(s)=P(s')  \wedge  |t|=K+n-m ] \Rightarrow [f(x_0',s't)\in X_S]\nonumber.
		\end{align}
		Since $K+n-m\ge K+1$, we know that $(K+1)$-step instant pre-opacity is violated, which is a contradiction.
	\end{proof}
	
	\begin{figure}
		\centering
		\includegraphics[width=0.23\textwidth]{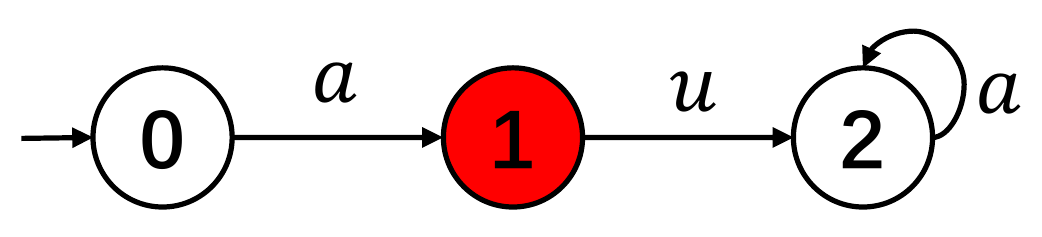} 
		\caption{A system that is current-state opaque but is not $0$-step instant pre-opaque, where $X_S=\{1\}$ and $E_{uo}=\{u\}$.} \label{fig:G_3}
	\end{figure}

	One may conjecture that  $0$-step instant  pre-opacity is equivalent to current-state opacity. 
	However, it is not exactly the case. 
	For example, let us consider system $G_3$ shown in Figure~\ref{fig:G_3}. 
	This system is current-state opaque as the intruder cannot distinguish states $1$ and $2$ after observing $a$ due to unobservable event $u$. 
	On the other hand, it is not $0$-step instant pre-opaque according to our definition since the intruder can predict one step ahead for sure that the system will reach the secret state when observing nothing. This difference is due to the fact that we consider instant in terms of actual event steps rather than the observation steps.   
	The only conclusion we can draw is that current-state opacity is weaker than $0$-step instant  pre-opacity, which is stated as follows.
	
	\begin{mypro} 
		If $G$ is  $0$-step instant/trajectory pre-opaque, then $G$ is current-state opaque. 
	\end{mypro}
	\begin{proof}
		It suffices to show that  $0$-step instant pre-opacity implies current-state opacity since  $0$-step trajectory  pre-opacity is stronger than $0$-step instant pre-opacity.
		Suppose  $G$ is   $0$-step instant pre-opaque. 
		Let us consider arbitrary initial state $x_0\in X_0$ and string $s\in \mathcal{L}(G,x_0)$. 
		Note that for string $s$, we can always find $\hat{s}\in  \mathcal{L}_o(G,x_0)$, by removing the unobservable tail (if any) of $s$ such that $P(s)=P(\hat{s})$. 
		Since $G$ is   $0$-step instant pre-opaque, by setting $n$ in Definition~\ref{def:ins} as $n=0$, we know that 
		there exist $ x_0'\in X_0$ and $s'\in \mathcal{L}_o(G,x'_0))$ such that $P(s')=P(\hat{s})=P(s)$ and $ f(x_0',s')\notin X_S$. 
		This means that the system is current-state opaque.
	\end{proof} 
	
	Based on the above discussion, we  summarize the relationships among the proposed notions of pre-opacity and existing notions of opacity in Figure~\ref{fig:relation}.
	
	\begin{remark} 
		Finally, we note that the proposed concept of pre-opacity is also related to the notion of fault predictability or fault prognosability in the literature \cite{jeron2008pred,genc2009predictability,kumar2010decentralized}, which captures whether or not a fault event can always be predicted unambiguously a certain number of steps ahead before it actually occurs. Conceptually, by considering the visit of secret states as fault, trajectory pre-opacity can be viewed as a dual problem of predictability. However,  trajectory pre-opacity is not exactly non-predictability. The former requires that \emph{all} secret paths cannot be predicted, while the latter says \emph{some} fault path cannot be predicted. Furthermore, our notion of instant opacity is quite different from predictability as we need to determine the specific instant of being secret; this issue does not occur in predictability analysis.
	\end{remark} 
	
	\section{Verification of Pre-Opacity}\label{sec:4}
	
	In this section, we show how to verify the proposed notions of pre-opacity. 
	Specifically, we present two state-based necessary and sufficient conditions for  $K$-step instant  pre-opacity and $K$-step trajectory  pre-opacity, respectively, that can be checked using the observer structure. Then we discuss the complexity of the verification problems.
	
	\subsection{Necessary and Sufficient Condition for  Instant  Pre-Opacity}
	Recall that a system is not $K$-step instant pre-opaque if after some observation, 
	each  possible state (immediately after the observation) will visit a secret state  in exactly $n$ steps for some $n\geq K$.  
	This suggests that $K$-step instant pre-opacity can be checked by combining the current-state estimation together with the reachability analysis. 
	To this end, we  further  introduce some necessary notions. 
	
	We say that a state $x\in X$ is a \emph{$K$-step indicator} state if it will reach a secret state inevitably in exactly $K$ steps, i.e.,  
	\[
	R_K(x)\subseteq X_S.
	\]
	For any $K\in \mathbb{N}$, we define
	\[
	\Im_K:=\{x\in X:  R_K(x)\subseteq X_S  \} \subseteq X
	\]
	as the set of $K$-step indicator states.
	
	Then the following theorem shows  that $K$-step instant  pre-opacity can be simply characterized in terms of current-state estimate and $K$-step indicator states. 
	\begin{mythm}\label{thm:strong-basic}
		System $G$ is $K$-step instant pre-opaque if and only if 
		\[
		\forall \alpha\in P(\mathcal{L}(G)), \forall n\geq K: \hat{\mathcal{E}}(\alpha)\not\subseteq \Im_n.
		\]
	\end{mythm}
	\begin{proof}
		($\Rightarrow$) 
		By contraposition. 
		Suppose that there exists  a string $\alpha\in P(\mathcal{L}(G))$ and an integer $n\geq K$ such that  $\hat{\mathcal{E}}(\alpha)\subseteq \Im_n$. 
		Let us consider
		an initial state $x_0\in X_0$ and a string $s\in \mathcal{L}_o(G,x_0)$ such that $P(s)=\alpha$. 
		Since $\hat{\mathcal{E}}(\alpha)\subseteq \Im_n$,  for any initial state $x_0'\in X_0$ and   string $s'\in \mathcal{L}_o(G,x_0)$ such that $P(s')=\alpha$, we have    $f(x_0',s')\in \Im_n$, i.e., $R_n(f(x_0',s'))\subseteq X_S$ 
		This means that for  
		for any $t\in \mathcal{L}(G,f(x_0',s'))$ and $|t|=n$, we have $f(x_0',s't)\in X_S$.  
		This means that system $G$ is not $K$-step instant  pre-opaque.
		
		($\Leftarrow$) 
		Still by contraposition. 
		Suppose that $G$ is not $K$-step instant  pre-opaque, which means that
		there exists an initial state $x_0\in X_0$, a string $s\in \mathcal{L}_o(G,x_0)$ and an integer $n\geq K$ such that 
		\begin{align}
		&(\forall x_0'\in X_0)(\forall s'\in \mathcal{L}_o(G,x'_0), s't\in \mathcal{L}(G,x'_0)) \nonumber\\
		&[P(s)=P(s')  \wedge  |t|=n  \Rightarrow f(x_0',s't)\in X_S]\nonumber
		\end{align}
		Then let $\alpha=P(s)$. Clearly, we have   $R_{n}(\hat{{\mathcal{E}}}(\alpha))\subseteq X_S$, i.e., $\hat{{\mathcal{E}}}(\alpha)\subseteq \Im_{n}$.
		This violates the condition in the theorem.
	\end{proof}
	
	Theorem~\ref{thm:strong-basic} essentially provides a state-based characterization of the language-based definition of $K$-step instant pre-opacity. However, it still cannot be directly used for the verification of instant pre-opacity. The main issue is that we need to check whether or not $\hat{\mathcal{E}}(\alpha)\not\subseteq \Im_n$ for any $n\geq K$, which has infinite number of instants. The following result further generalizes  Theorem~\ref{thm:strong-basic} and shows that it suffices to check $\hat{\mathcal{E}}(\alpha)\not\subseteq \Im_n$  for a bounded number of instants. 
	
	\begin{mypro}\label{prop:bound}
		For any $\alpha\in P(\mathcal{L}(G))$, the following two statements are equivalent:
		\begin{enumerate}[(i)]
			\item 
			$\forall n\geq K: \hat{\mathcal{E}}(\alpha)\not\subseteq \Im_n$; 
			\item 
			$\forall n\in \{K,K+1, \dots, K+2^{|X|}-1\}: \hat{\mathcal{E}}(\alpha)\not\subseteq \Im_n$.
		\end{enumerate}
	\end{mypro}
	\begin{proof}(i)$\Rightarrow$(ii) is trivial. 
		Hereafter, we show that (ii)$\Rightarrow$(i). 
		Let $q:=\hat{{\mathcal{E}}}(\alpha)$ and we consider the reachable set of $q$  for each instant between $K$ and $K+2^{|X|}-1$, i.e., 
		$R_K(q), R_{K+1}(q),\dots,R_{K+2^{|X|}-1}(q)$. 	
		For any $n\in \{K,\dots, K+2^{|X|}-1\}$, since $q\not\subseteq \Im_n$,  we know that there exists $x\in q$ such that $x\notin \Im_n$, i.e., $R_n(x)\not\subseteq X_S$. 
		Since $R_n(q) = \bigcup_{x\in q}R_n(x)$, we know that $R_n(q)\not\subseteq X_S$ for any $n\in \{K,\dots, K+2^{|X|}-1\}$.
		
		Now we note that  set $\{R_i(q):K\le i\le K+2^{|X|}-1\} \subseteq 2^X$ is non-empty, so it contains at most $2^{|X|}$ elements. 
		Therefore, there must exist two instants   $K\le i<j\le K+2^{|X|}-1$ such that $R_{i}(q)=R_{j}(q)$.
		Furthermore, by the definition of $K$-step reachable set, we have 
		\[
		R_{n+k}(q)=R_n( R_k(q) )=R_k(R_n(q))
		\]
		Then for any instant $n'> K+2^{|X|}-1$, we can always write it in the form of 
		\[
		n'=  i+ (j-i)\times k+ m 
		\]
		where  $1\leq k,0\leq m<(j-i)$ are two   integers. Furthermore, since $R_{i}(q)=R_{j}(q)$, we have $R_{i}(q)=R_{i+ (j-i)\times k}(q)$ for any $k\geq 0$, so
		\[
		R_{n'}(q)=R_{m}( R_{i+ (j-i)\times k}(q)  )=R_{m+i}(q). 
		\]
		However, since $m<j-i$, we have $m+i<j$. Therefore,  
		\[
		\forall n'> K+2^{|X|}-1: R_{n'}(q)=R_{m+i}(q)\not\subseteq X_S.
		\]
		This further implies that 
		\[
		\forall n'> K+2^{|X|}-1: q\not\subseteq \Im_{n'}, 
		\]
		which completes the proof.
	\end{proof}
	
	One may ask why we need to search for the entire next $2^{|X|}$ instants to obtain the upper bound in Proposition~\ref{prop:bound}. 
	However, this upper bound seems to be unavoidable. 
	To see this, let us consider the system shown in Figure~\ref{fig:Gexm}, where all events are unobservable and red states denote secret states. This system is not $K$-step instant pre-opaque for any $K$ since one can determine for sure (by observing nothing) that the system will be at a secret state for instants $k\cdot 30,k=1,2\dots$, where $30$ is the least common multiple of cycle lengths $2,3$ and $5$. 
	Therefore, the first violation of  $\hat{\mathcal{E}}(\alpha)\not\subseteq \Im_n$ occurs at $n=30$. 
	Similarly, one could add more states to create more such cycles and the upper bound for searching $\Im_n$ will grow exponentially.
	However, this exponentially searching bound is only needed when the system contains unobservable events. 
	In the following result, we show that such an upper bounded search can be avoided for the extreme case when there is no unobservable event in the system.

	\begin{mypro}\label{prop:nobound}
		Under the assumption that all events in $G$ are observable, then $G$ is $K$-step instant pre-opaque if and only if 
		\[
		\forall \alpha\in P(\mathcal{L}(G)): \hat{\mathcal{E}}(\alpha)\not\subseteq \Im_K.
		\]
	\end{mypro} 
	\begin{proof}The necessity follows directly from Theorem~\ref{thm:strong-basic}.  
		To show the sufficiency,  suppose that $\forall \alpha\in P(\mathcal{L}(G)): \hat{\mathcal{E}}(\alpha)\not\subseteq \Im_K$ 
		and assume that $G$ is not $K$-step instant pre-opaque. 
		Then by Theorem~\ref{thm:strong-basic}, we know that 
		there exist $\alpha\in P(\mathcal{L}(G))$ and $n> K$ such that $\hat{\mathcal{E}}(\alpha)\subseteq \Im_n$.
		In other words, we have that
		\begin{align}
		&(\forall x_0\in X_0)(\forall s\in \mathcal{L}_o(G,x_0), st\in \mathcal{L}(G,x_0)) \nonumber\\
		&[P(s)=\alpha  \wedge  |t|=n]  \Rightarrow [f(x_0,st)\in X_S]\nonumber
		\end{align}
		For any $t$ satisfying above condition, we let $t=t_1t_2$, where $|t_1|=n-K$ and $|t_2|=K$. Then we know that $R_K(\hat{\mathcal{E}}(P(st_1)))\subseteq X_S$, i.e., $\hat{\mathcal{E}}(P(st_1))\subseteq \Im_K$, which is a contradiction.
	\end{proof}

	\begin{figure}
		\center			\includegraphics[width=0.31\textwidth]{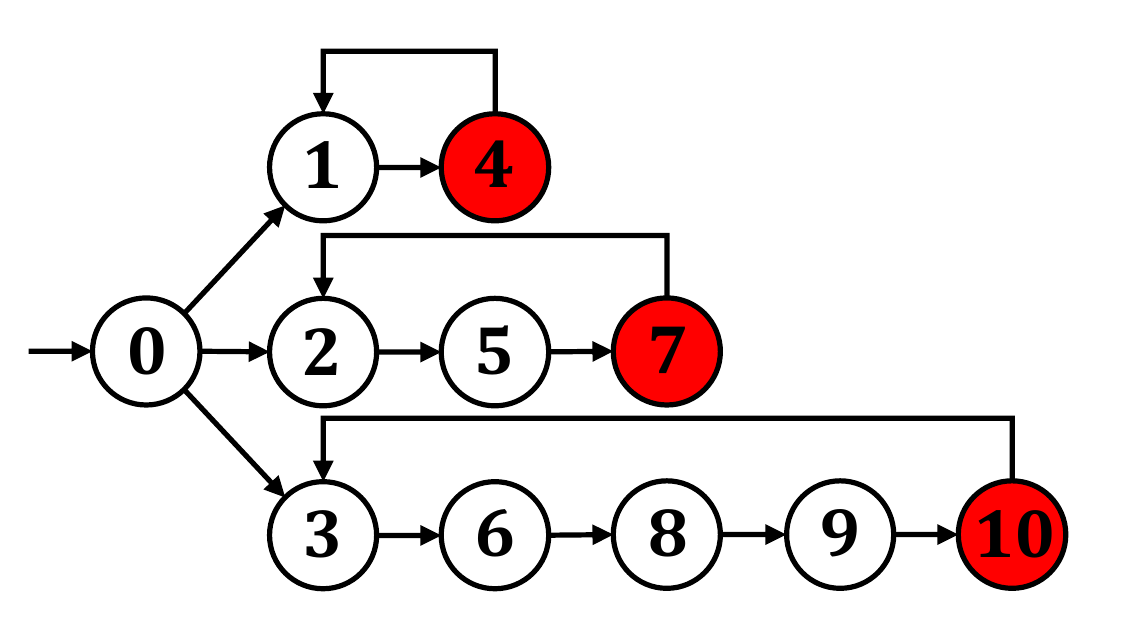}
		\caption{A system  where all events are unobservable and red states denote secret states.}\label{fig:Gexm}
	\end{figure}

	\subsection{Necessary and Sufficient Condition for Trajectory  Pre-Opacity}
	
	Now we discuss the case of $K$-step trajectory pre-opacity. 
	First, we say that a state $x\in X$ is a  \emph{non-indicator state} if there exists an infinitely long string defined from this state along which no secret state is visited. 
	Formally, we define the set of non-indicator states by 
	\begin{align}
	\mathcal{N}:=
	\left\{x\in X: 
	\begin{array}{c c}
	(\forall n\geq 0)(\exists s\in \mathcal{L}(G,x): |s|>n)\\
	(\forall t\in \overline{\{s\}})[ f(x,t)\not\in X_S ]
	\end{array}
	\right\}
	\end{align}
	Since the number of states in $G$ is finite, a state is a non-indicator state if and only it can 
	reach  a cycle, in which all states are non-secret,  via a sequence of non-secret states. 
	In other words, if a state is not in $\mathcal{N}$, then it is an indicator state in the sense that a secret state will be visited inevitably  from this state.  
	
	\begin{remark} 
		Note that a state is not a non-indicator state does not necessarily imply that it is a $K$-step indicator state for some $K$ since the latter requires the system to be at a secret state for some specific instant while the former does not require this information.
		Furthermore, a state is an indicator state does not imply that any state reachable from this state is an indicator state. This is because after passing through a secret state, the status of indicating may become non-indicating. 
	\end{remark}

	Therefore, if the system is at a state whose $K$-step reachable set is a subset of  indicator state, then based on this state information, one can predict $K$-step ahead that a secret state will be visited. 
	We define 
	\[
	\mathcal{N}_{K}:=\{x\in X: R_K(x)\cap \mathcal{N}\not=\emptyset\}\subseteq X
	\]
	as the set of  states the intruder cannot make such a prediction. 
	Then we have the following theorem.

	\begin{mythm}\label{thm:weak}
		System $G$ is $K$-step trajectory pre-opaque if and only if 
		\[
		\forall \alpha\in P(\mathcal{L}(G)):\hat{\mathcal{E}}(\alpha)\cap \mathcal{N}_{K}\not=\emptyset.
		\]
	\end{mythm}
	
	\begin{proof}
		($\Rightarrow$) By contradiction. 
		Suppose that $G$ is $K$-step trajectory  pre-opaque and assume that there exists a string $\alpha\in P(\mathcal{L}(G))$ such that $\hat{\mathcal{E}}(\alpha)\cap \mathcal{N}_{K}=\emptyset$ holds. 
		According to Definition~\ref{def:traj}, for any $n\ge K$, there exists $x'_0\in X_0,s'\in \mathcal{L}_o(G,x'_0), s't_1t_2\in \mathcal{L}(G,x'_0)$ such that $P(s')=P(s)=\alpha, |t_1|= K$, $|t_2|= n-K$ and for any $w\in \overline{\{t_2\}}$, we have $f(x'_0,s't_1w)\not\in X_S$. 
		Now, let us choose $n$ such that $n\ge |X|+K$, i.e., $|t_2|\geq |X|$. 
		Since $f(x'_0,s't_1t_2)$ can pass through at most $|X|$ states, 
		there are at least two repeated states that forms a cycle along the path of   
		$t_2$ starting from  $f(x'_0,s't_1)$. 
		This immediately implies that $f(x'_0,s't_1)\in \mathcal{N}$. 
		Furthermore, we have 
		$f(x'_0,s't_1)\in R_K(f(x'_0,s'))$  since  $|t_1|=K$.
		Therefore, we have 
		$R_K(f(x'_0,s'))  \cap \mathcal{N} \not=\emptyset$, i.e., $f(x'_0,s')\in\mathcal{N}_{K}$.
		Since 
		$P(s')=\alpha$ and $s'\in E^*E_o\cup\{\epsilon\}$, we have 
		$f(x'_0,s')\in \hat{\mathcal{E}}(\alpha)$, which implies that 
		$\hat{\mathcal{E}}(\alpha)\cap \mathcal{N}_{K}\not =\emptyset$. 
		This, however, contradicts to our assumption.
		
		($\Leftarrow$) By contradiction. 
		Suppose that for any $\alpha\in P(\mathcal{L}(G))$, we have $\hat{\mathcal{E}}(\alpha)\cap \mathcal{N}_{K}\not=\emptyset$ and assume that $G$ is not $K$-step trajectory  pre-opaque, i.e., 
		there exist a state $x_0\in X_0$, a string $s\in \mathcal{L}_o(G,x_0)$ and an integer $n \ge K$ such that 
		\begin{align}\label{eq;pf}
		&(\forall x_0'\in X_0)(\forall s'\in \mathcal{L}_o(G,x'_0), \forall t_1t_2\in \mathcal{L}(G,f(x_0',s'))\text{ s.t. } \nonumber\\
		&[P(s)=P(s')  \wedge   |t_1|=K \wedge   [|t_1t_2|=n]   \nonumber\\
		&\Rightarrow [\exists w\!\in\! \overline{\{t_2\}}:f(x_0',s't_1w)\!\in\! X_S] 
		\end{align}
		Let us consider an arbitrary state $x$ in $\hat{\mathcal{E}}(P(s))$. 
		This means that there exist 
		a state $x_0'\in X_0$ and a string $s'\in \mathcal{L}_o(G,x'_0)$ 
		such that $f(x_0',s')=x$ and $P(s')=P(s)$. 
		However, according to Equation~\eqref{eq;pf}, 
		any string of length $n$ from state $x$ must pass through a secret state between its $K$th instant and its $n$th instant. 
		This means that 
		$R_K(x) \cap \mathcal{N}=\emptyset$, i.e., $x\notin \mathcal{N}_K$. 
		Note that $x$ is an arbitrary state in  $\hat{\mathcal{E}}(P(s))$. 
		Therefore, we have 
		$\hat{\mathcal{E}}(P(s)) \cap \mathcal{N}_K=\emptyset$. 
		However, this contradicts to our assumption that  $\hat{\mathcal{E}}(\alpha)\cap \mathcal{N}_{K}\not=\emptyset$ for any $\alpha\in P(\mathcal{L}(G))$.
	\end{proof}
	
	\subsection{Verification Algorithms}
	
	Now, let us discuss how to use the derived necessary and sufficient conditions to verify $K$-step instant or trajectory pre-opacity. 
	To this end, we need to compute
	\begin{itemize}
		\item 
		All possible state estimates, i.e.,  $\{\hat{\mathcal{E}}(\alpha):\alpha\in P(\mathcal{L}(G))\}$; 
		\item 
		A set of $n$-step indicator states for $K\leq n\leq K+2^{|X|}-1$, i.e.,  $\{\Im_K,\dots, \Im_{K+2^{|X|}-1}\}$ (for instant pre-opacity); 
		\item
		The set of states whose $K$-step reachable set contains at least a non-indicator state, i.e., 
		$\mathcal{N}_K$ (for trajectory pre-opacity). 
	\end{itemize}
	
	\subsubsection{Computation of $\hat{\mathcal{E}}(\alpha)$}Note that, compared with the standard current-state estimate, the state estimate considered here does not contain the unobservable tail. This can be computed by a slightly modified version of the standard observer automaton (we still call it observer here for the sake of simplicity). 
	Formally, the observer of $G$ is a new DFA   
	\[
	Obs(G)=(Q_{obs},E_o,f_{obs},q_{obs,0}), 
	\]
	where 
	$Q_{obs}\subseteq 2^{X}\setminus \emptyset$ is the set of states, 
	$E_o$ is the set of events, 
	$q_{obs,0}=X_0$ is the initial state, and 
	$f_{obs}:Q_{obs}\times E_o\to Q_{obs} $ is the deterministic transition function defined by: 
	for any $q\in Q_{obs},\sigma\in E_o$, we have 
	\begin{equation}
	f_{obs}(q,\sigma)=\{ x\!\in\! X: \exists x'\!\in\! q, w\!\in\! E_{uo}^*\text{ s.t. }f(x',w\sigma)=x   \}
	\end{equation}
	For the sake of simplicity, we only consider the reachable part of the observer.  
	Then we have 
	\[
     \forall \alpha\in P(\mathcal{L}(G)): f(q_{obs,0},\alpha)=\hat{\mathcal{E}}(\alpha).
     \]
	Therefore, all possible state estimate $\hat{\mathcal{E}}(\alpha)$ can be computed with complexity $O(|E_o|2^{|X|})$. 
	\subsubsection{Computation of $\Im_n$} 
	For any give $n\geq 0$, one can compute $\Im_n$ by backtracking $n$ steps from the set of all secret states. 
	Formally, one can define an operator 
	$F: 2^X\to 2^X$ by: for any $q\in 2^X$, we have 
	\begin{equation} 
	F(q)=\{x\in X: \forall \sigma\in E, f(x,\sigma)!\text{ s.t. }f(x,\sigma)\in q    \}.
	\end{equation}
	Then one can easily check that 
	\[\Im_n= F^n(\Im_{0})  \text{ with }\Im_0=X_S\] which can be computed with complexity $O( n|E_o||X|  )$.

	\subsubsection{Computation of $\mathcal{N}_K$}
	To compute $\mathcal{N}_K$, first we need to compute the set of non-indicator states $\mathcal{N}$. 
	To this end, we can remove all secret states in $G$ and compute all strongly  connected components, i.e., cycles; this can be done by, e.g.,  Kosaraju's algorithm with a linear complexity in the size of $G$ \cite{algorithm}. Then those states that can reach a non-secret cycle are the set of non-indicator states. 
	Therefore, computing set $\mathcal{N}$ can be done in $O(|E||X|)$.
	In order to compute $\mathcal{N}_K$, one can backtrack from $\mathcal{N}$ using another operator 
	$W: 2^X\to 2^X$ defined by: for any $q\in 2^X$, we have 
	\begin{equation} 
	W(q)=\{x\in X: \exists \sigma\in E\text{ s.t. }f(x,\sigma)\in q    \}.
	\end{equation}
	Then one can easily check that 
	\[
	\mathcal{N}_K= W^K(\mathcal{N}_0)\text{ with }\mathcal{N}_0=\mathcal{N}
	\]
	which can be computed with complexity $O( K|E_o||X|  )$. 
	Therefore, the overall complexity for computing set  $\mathcal{N}_K$ is 
	$O( K|E_o||X|   )$.
	
	Based on the above discussions, we summarize the algorithms for the verification of $K$-step instant pre-opacity and $K$-step trajectory pre-opacity by Algorithm~\textsc{Ins-Pre-Opa-Ver} and Algorithm~\textsc{Traj-Pre-Opa-Ver}, respectively.  
	The complexity of Algorithm~\textsc{Ins-Pre-Opa-Ver} is 
	$O( |E_o|2^{|X|}[K+(K+1)+\cdots+(K+2^{|X|}-1) ]|E_o||X|)=O(|E_o|^2|X|(K+2^{|X|-1})2^{2|X|} )$ for the general case 
	and is $O(K|E_o|^2|X|2^{|X|} )$ under the assumption that there is no unobservable event.
	The complexity of Algorithm~\textsc{Traj-Pre-Opa-Ver} is simply
	$O( K|E_o|^2|X|2^{|X|} )$, which is dominated by the size of the observer.  
	We illustrate the verification algorithms by the following examples. 
	
	\begin{algorithm}	
		\SetAlgoNoLine 
		\SetKwInOut{Input}{\textbf{input}}\SetKwInOut{Output}{\textbf{output}}	
		\Input{
			System $G$ with $X_S$,  $E_o$ and  $K$
		}
		\Output{ \textsc{Yes} or \textsc{No} }
		\BlankLine
		Construct the observer $obs(G)$\;
		\eIf{there is no unobservable event in $G$}
		{
			$M\gets 0$\;
		}
		{
			$M\gets  2^{|X|}-1$\;	
		}
		\For{$q\in Q_{obs}$} 
		{
			\For{$n\in \{K,K+1,\dots, K+M\}$} 
			{
				\If{$q\subseteq \Im_n$}
				{
					\textbf{return} \textsc{No}\;
				}
			} 
		}
		\textbf{return} \textsc{Yes}\;
		\caption{\textsc{Ins-Pre-Opa-Ver}\label{al:instant}}
	\end{algorithm}

	\begin{algorithm}	
		\SetAlgoNoLine 
		\SetKwInOut{Input}{\textbf{input}}\SetKwInOut{Output}{\textbf{output}}	
		\Input{
			System $G$ with $X_S$,  $E_o$ and  $K$
		}
		\Output{ \textsc{Yes} or \textsc{No} }
		\BlankLine
		Construct the observer $obs(G)$\;
		\For{$q\in Q_{obs}$} 
		{
			\If{$q \cap \mathcal{N}_K= \emptyset$}
			{
				\textbf{return} \textsc{No}\;
			}
		}
		\textbf{return} \textsc{Yes}\;
		\caption{\textsc{Traj-Pre-Opa-Ver}\label{al:trajectory}}
	\end{algorithm}
	\begin{myexm}
		Let us consider again system $G_1$ shown  in Figure~\ref{fig:G1} and we verify whether or not it is $K$-step trajectory pre-opaque. 
		First, we build its observer $Obs(G_1)$ as shown in Figure~\ref{fig:obs1}. 
		The only non-indicator state is $6$, i.e.,  $\mathcal{N}=\{6\}$. 
		For  $K=2$, we have $\mathcal{N}_2=W^2( \{6\}  )=\{2,4,5,6,7\}$. 
		Since $\{0\}\cap \{2,4,5,6,7\}=\emptyset$, we know that $G_1$ is not $2$-step trajectory pre-opaque. 
		However, for $K=3$, we have $\mathcal{N}_3=W^3( \{6\}  )=\{0,2,3,4,5,6,7\}$ and each observer state has a common element with $\mathcal{N}_3$. 
		Therefore, $G_1$ is   $3$-step trajectory pre-opaque.  
		These results are also consistent with our previous intuitive analysis. 
		
		However, this system is $K$-step instant pre-opaque for any $K\geq 0$. 
		To see this, it suffices to consider the case of $K=0$. In this case, we have  
		\begin{align}
		&\Im_0=\{4,7\}, \Im_1=\{2,5\}, \Im_2=\{3\}, \Im_3=\{1\}, \nonumber\\
		&\Im_4=\Im_5=\cdots=\emptyset \nonumber
		\end{align}
		Therefore, no observer state is a subset of any $\Im_i$, which implies $0$-step instant pre-opacity.
	\end{myexm}
	
	\begin{figure}
		\centering
		\subfigure[$Obs(G_1)$]{\label{fig:obs1}
			\includegraphics[width=0.43\textwidth]{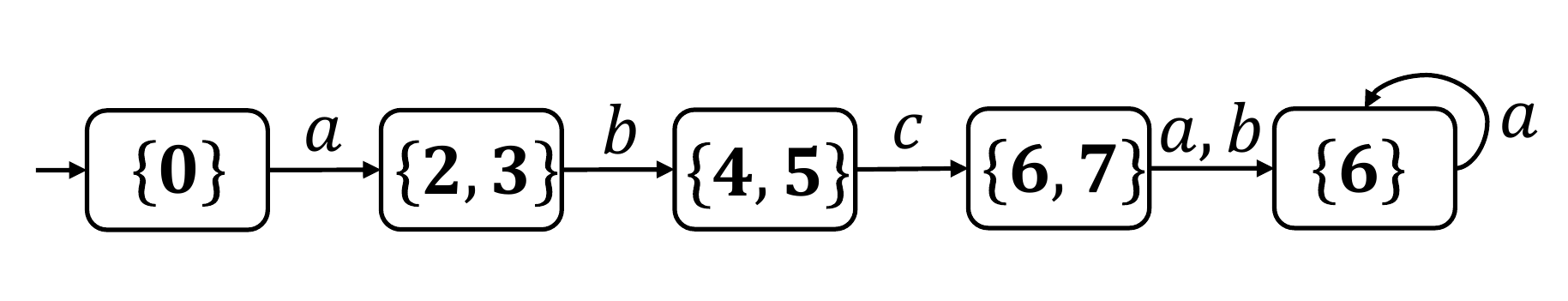}}
		\subfigure[$Obs(G_2)$]{\label{fig:obs2}
			\includegraphics[width=0.25\textwidth]{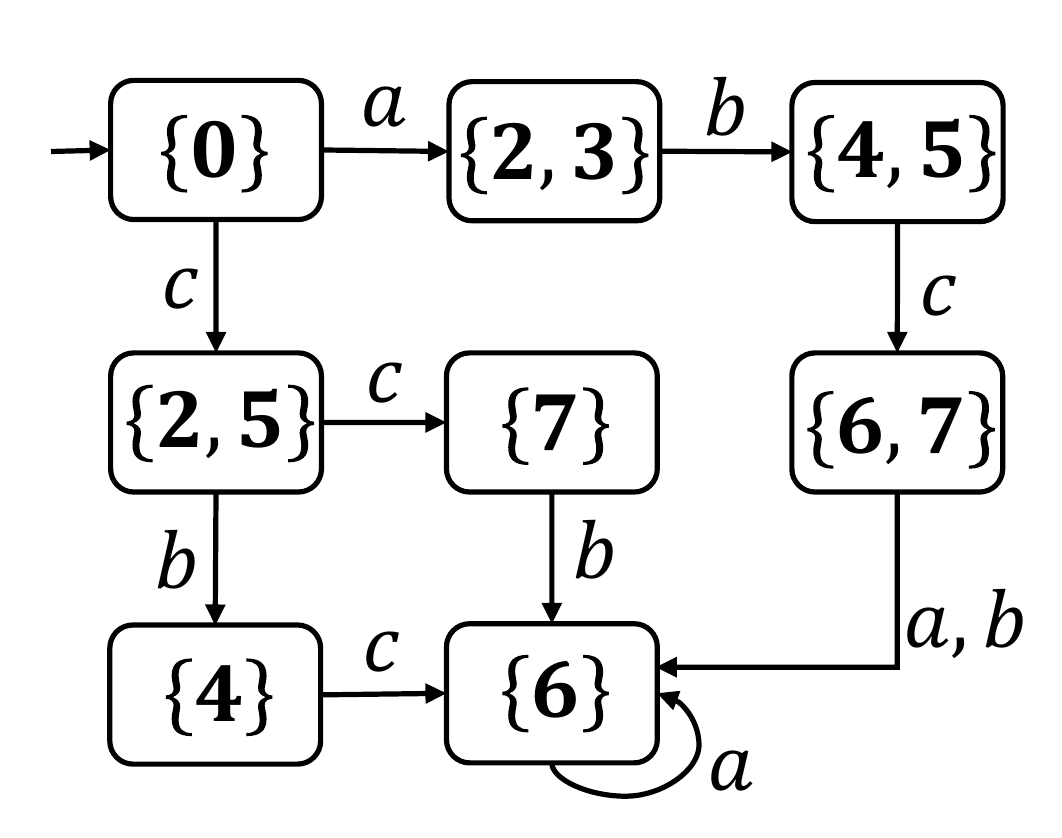}}
		\caption{Observers for $G_1$ and $G_2$, respectively.} \label{fig:obs}
	\end{figure}

	\begin{myexm}
		For system $G_2$ shown in Figure~\ref{fig:G2},  its observer is shown in Figure~\ref{fig:obs1}. 
		Then we have 
		\begin{align}
		&\Im_0=\{4,7\}, \Im_1=\{2,5\}, \Im_2=\{3\}, \Im_3=\Im_4=\cdots=\emptyset \nonumber
		\end{align}
		However, for observer state $\{2,5\}$, we have $\{2,5\}\subseteq \Im_1$, which means that 
		$G_2$ is not $1$-step instant pre-opaque. 
		On the other hand, $G_2$ is $K$-step instant pre-opaque for any $K\ge 2$ as no state in $Obs(G_2)$ is a subset of any $\Im_n,n\geq 2$.
	\end{myexm}

	\subsection{The Complexity of $K$-Step Pre-Opacity}
	
	Note that the complexity of Algorithm~\textsc{Ins-Pre-Opa-Ver} and Algorithm~\textsc{Traj-Pre-Opa-Ver} are both exponential in the number of states in $G$. 
	Next, we show that both properties are essentially PSPACE-hard; therefore, the exponential complexity seems to be unavoidable.

	\begin{mythm}
		Deciding whether or not $G$ is $K$-step instant (or trajectory) pre-opaque is PSPACE-hard even when $G$ is deterministic.
	\end{mythm}
	
	\begin{proof}
		Given two non-deterministic automata (NFAs) $G_1=(X_1,E,f_1,X_{1,0})$ and $G_2=(X_2,E,f_2,X_{2,0})$, 
		the problem of language containment asks to decide whether or not $\mathcal{L}(G_1)\subseteq \mathcal{L}(G_2)$. This problem is known to be  PSPACE-hard.	
		Hereafter, we show that checking $K$-step instant/trajectory pre-opacity  is also PSPACE-hard
		by reducing the language containment problem to the  pre-opacity verification problem.
		
		Let $G_1=(X_1,E,f_1,X_{1,0})$ and $G_2=(X_2,E,f_2,X_{2,0})$ be two NFAs with initial states $X_{1,0}$ and $X_{2,0}$, respectively. Without loss of generality, we   assume $G_1$ and $G_2$ are live;  otherwise, we can add a self-loop with a new event at each state in $G_1$ and $G_2$.
		Note that, in the analysis of pre-opacity, we assume that the transition function is deterministic; this gap can be bridged by using unobservable events to mimic non-determinism. 
		Formally, let $E_u=\{u_1,u_2,\dots, u_k\}$ be a set of new unobservable events. 
		Then for each NFA $G_i$, we construct a new DFA $\tilde{G}_i=(\tilde{X}_i,\tilde{E},\tilde{f}_i,\tilde{X}_{i,0})$ by: 
		$\tilde{X}_i= {X}_i\cup\{(x,\sigma)\in X_i\times E :f_i(x,\sigma )! \}$, 
		$\tilde{E}=E\cup E_u$, $X_{i,0}=\tilde{X}_{i,0}$, and 
		$\tilde{f}_i: \tilde{X}_i\times \tilde{E}\to \tilde{X}_i$ is the deterministic transition function defined by: for any $f_i(x,\sigma)!$, 
		we have $\tilde{f}_i(x,\sigma)=(x,\sigma)$ and $f_i(x,\sigma)=\{ \tilde{f}_i( (x,\sigma)  ,u):  u\in E_u\}$. 
		The construction of $\tilde{G}_i$ is illustrated by Figure~\ref{fig:illu}. 
		Clearly, one has  $\mathcal{L}(G_1)\subseteq \mathcal{L}(G_2)$ iff  $P(\mathcal{L}(\tilde{G}_1))\subseteq P(\mathcal{L}(\tilde{G}_2))$.

		Now we construct a new DFA $\tilde{G}=(\tilde{X},\tilde{E},\tilde{f},\tilde{X}_{0})$ by taking the union of $\tilde{G}_1$ and $\tilde{G}_2$, i.e.,
		$\tilde{X}=\tilde{X}_1\cup \tilde{X}_2$, $\tilde{f}$ is consistent with $\tilde{f}_1$ and $\tilde{f}_2$, and $\tilde{X}_{0}=\tilde{X}_{1,0}\cup \tilde{X}_{2,0}$. 
		Then, for system $\tilde{G}$, we let $X_S=X_1$ and $E_u$ be the set of unobservable events. 
		We show that $\tilde{G}$ is $0$-step instant (or trajectory) pre-opaque if and only if
		$\mathcal{L}(G_1)\subseteq \mathcal{L}(G_2)$. 
		
		($\Rightarrow$)
		To see this, we suppose that $\mathcal{L}(G_1)\not\subseteq \mathcal{L}(G_2)$,  then we know that there exists a string $s\in \mathcal{L}(G_1)\setminus \mathcal{L}(G_2)$, i.e.,  there exists a string $t\in P(\mathcal{L}(\tilde{G}_1))\setminus P(\mathcal{L}(\tilde{G}_2))$, since $\mathcal{L}(G_1)\subseteq \mathcal{L}(G_2)$ is equivalent to  $P(\mathcal{L}(\tilde{G}_1))\subseteq P(\mathcal{L}(\tilde{G}_2))$.
		Therefore, after observing $t$ in $\tilde{G}$, since $X_S=X_1$, we know   for sure that the system now is at a secret state and will be at secret states for any future instant. 
		Hence, $\tilde{G}$ is not $0$-step instant (or trajectory) pre-opaque.
		
		($\Leftarrow$) 
		Suppose that $\mathcal{L}(G_1)\subseteq \mathcal{L}(G_2)$ and we assume that, for the sake of contradiction, 
		$\tilde{G}$ is not $0$-step trajectory pre-opaque, which means it is also not   $0$-step instant pre-opaque.
		Then we know that there exists a string $s\in P(\mathcal{L}(\tilde{G}))$ such that $\hat{\mathcal{E}}(s)\cap \mathcal{N}_0=\emptyset$. 
		Note that we have $P(\mathcal{L}(\tilde{G}_1))\subseteq P(\mathcal{L}(\tilde{G}_2))$.  Since $\mathcal{L}(G_1)\subseteq \mathcal{L}(G_2)$, this also implies that $P(\mathcal{L}(\tilde{G}))=P(\mathcal{L}(\tilde{G_2}))$ and $s\in P(\mathcal{L}(\tilde{G_2}))$.
		However, since every state in $\tilde{G_2}$ is non-secret and $\tilde{G_2}$ is live, we have $\tilde{X}_2\subseteq \mathcal{N}_0$. Therefore, it is not possible that $\hat{\mathcal{E}}(s)\cap \mathcal{N}_0=\emptyset$, which is a contradiction.
		
		Overall, we conclude that deciding whether or not $G$ is $K$-step instant/trajectory pre-opacity is PSPACE-hard.
	\end{proof}

	\begin{figure}
		\centering
			\includegraphics[width=0.46\textwidth]{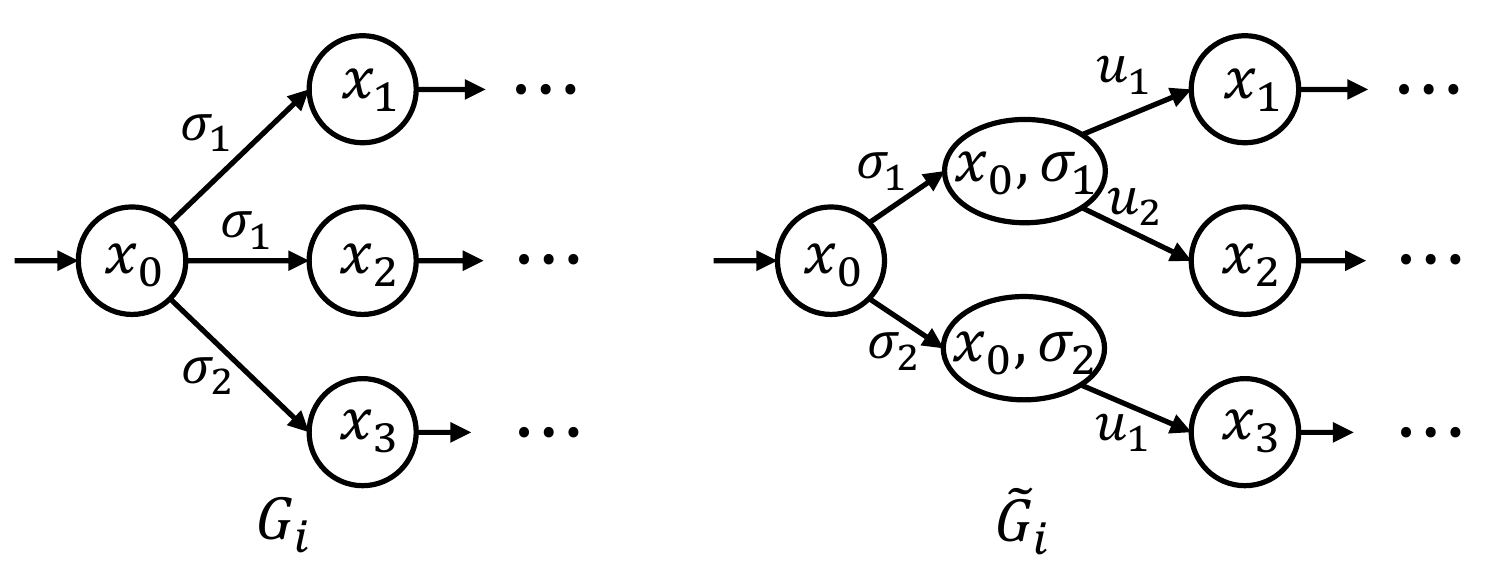}
		\caption{Conceptual illustration of how to construct $\tilde{G_i}$ from $G_i$}. \label{fig:illu}
	\end{figure} 
	
	\section{Secret Intention as a Sequence Pattern}\label{sec:5}
		\begin{figure*}
		\centering
		\subfigure[The topology of a factory]{\label{fig:case}
			\includegraphics[width=0.32\textwidth]{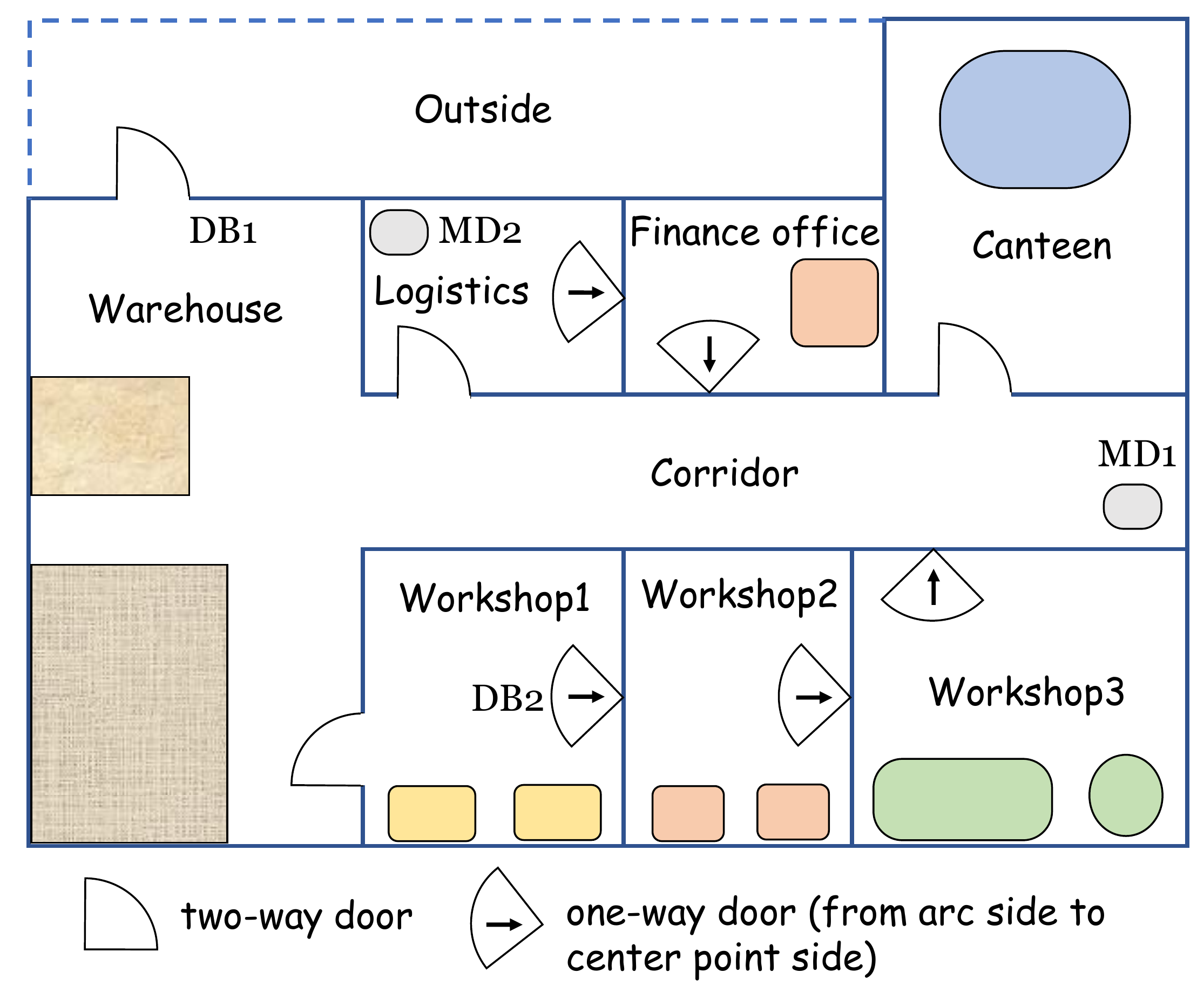}}
		\subfigure[The specification automaton $G_3$]{\label{fig:G3}
			\includegraphics[width=0.3\textwidth]{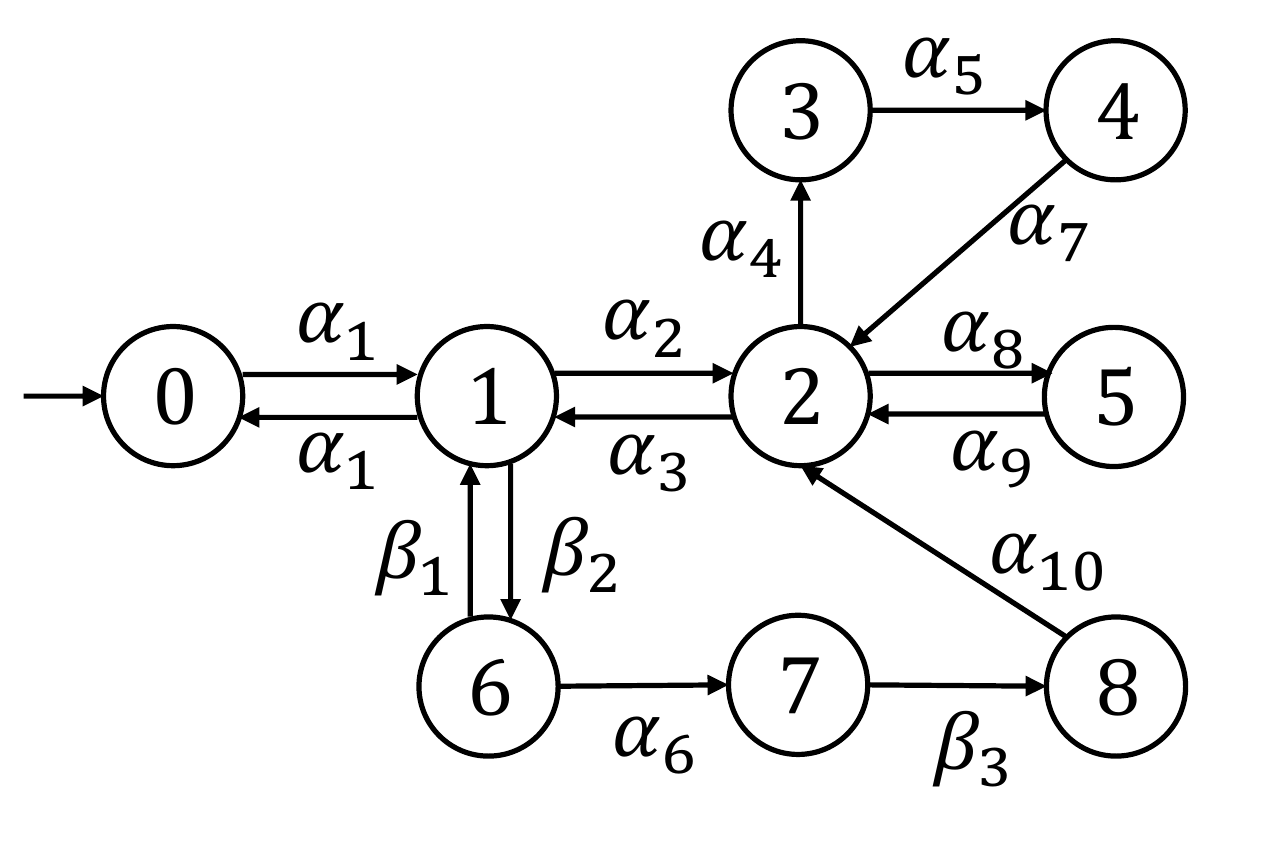}}
		\subfigure[$G_{\Omega}$]{\label{fig:marked}
			\includegraphics[width=0.32\textwidth]{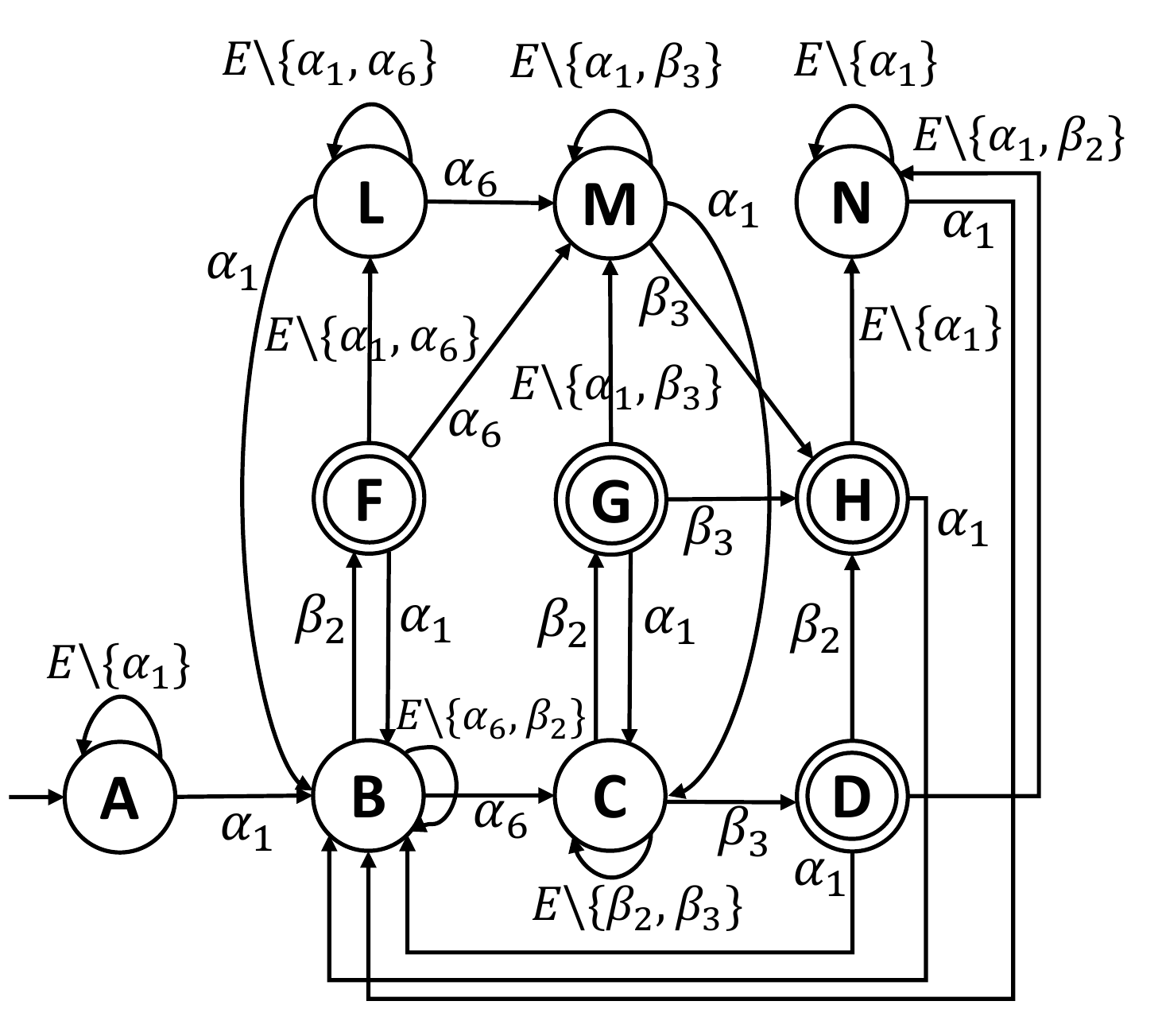}}
		
		\caption{An illustrative case of pattern pre-opacity.} \label{fig:exm2}
	\end{figure*}
	In the previous sections, we model the secret intention of the system as the willing to reach some secret states. In this section, we further generalize this setting by considering the secret intention as the willing to execute some particular sequences of events, which we call a \emph{sequence pattern}. 
	We present an illustrative example that motivates the definition of pattern pre-opacity and show how it can be reduced to state-based pre-opacity. 
	
	\subsection{Illustrative Example of Pattern Pre-Opacity}

    We consider a location-tracking/prediction type problem in a smart factory building equipped with sensors as shown in Figure~\ref{fig:case}.  
    The factory has eight regions of interest:  a \emph{warehouse}, a \emph{logistics}, a \emph{finance office}, a \emph{canteen}, a \emph{corridor} and three \emph{workshops}.  
    We assume there is a person in the factory  that can move from one region to another by passing through a door; some doors are one-way and some are two-way as depicted in the figure. 
	In particular, there are two doors $DB_1$ and $DB_2$ secured by door barrier sensors, which allow to observe if a person enters the corresponding rooms. 
	Furthermore,  there are two additional motion detector sensors ($MD_1$ and $MD_2$) at corridor and logistics, respectively; they can detect if a person moves to the corridor (or logistics).
	The building monitor is able to use these sensors to track and predict the behavior of the person. 
	
	According to the structure of factory and different types of doors, the overall system, which is the mobility of the person, can be modeled as a DES  as shown in Figure~\ref{fig:G3}, 
	where states $0$ to $8$ represent, respectively, regions \emph{outside}, \emph{warehouse}, \emph{corridor}, \emph{logistics}, \emph{finance office}, \emph{canteen}, \emph{workshop 1}, \emph{workshop 2} and \emph{workshop 3}.  
	Based on the distribution of motion detector sensors and door barrier sensors, we know that the set of observable events is 
	\[
	E_o=\{\alpha_1,\alpha_2,\alpha_3,\alpha_4,\alpha_5,\alpha_6,\alpha_7,\alpha_8,\alpha_9,\alpha_{10}\}.
	\]

	Now we assume that the person wants to move in the factory to complete two tasks ``secretly": 
	(\emph{task 1}) first goes to \emph{warehouse} and then goes to \emph{workshop 1}; 
	(\emph{task 2}) first goes to \emph{warehouse}, and then enters \emph{workshop 2}, and finally gets to \emph{workshop 3}. 
	Furthermore, the person wants to hide its intention for executing the above sequences against the monitor before they are completed. 
	In this setting, ``secret intention" is no longer visiting a secret state in the future. 
	Instead, completing any sequence containing string $\alpha_1\beta_2$ or $\alpha_1\alpha_6\beta_3$ can be viewed as a secret behavior. 	
	One can check that the person may not be able to hide  its intention to complete task 2 more than one step before its completion. 
	This is because once  motion detector sensor $DB_2$ is triggered, the building monitor can determine for sure that the person was from \emph{warehouse}, and is currently at \emph{workshop 2} and  will go to \emph{workshop 3} in one step to complete the task. 
	To formally describe this scenario, we propose $K$-step instant/trajectory pattern pre-opacity in the next part.

	\subsection{Definitions of Pattern Pre-Opacity} 
	Now, we formally formulate the notion of pattern pre-opacity.  
	Specifically, we consider a  \emph{sequence pattern}  modeled as a regular language $\Omega\subseteq E^*$ in order to describe the secret behaviors of the system.  
	Then we say that a system is $K$-step pattern pre-opaque if any completion of a string in the pattern can be predicted $K$-step ahead. 
	Depending on whether or not the intruder needs to determine the specific instant of the completion, 
	pattern pre-opacity can also be categorized as instant pre-opacity and trajectory pre-opacity.

	\begin{mydef}($K$-Step Instant Pattern Pre-Opacity)\label{def:pinst}
		Given system $G$,  set of observable events $E_o$,  sequence pattern $\Omega$, and non-negative integer $K\in \mathbb{N}$, system $G$ is said to be $K$-step instant pattern pre-opaque (w.r.t.\ $E_o$ and $\Omega$) if 
		\begin{align}
		&(\forall x_0\in X_0)(\forall s\in \mathcal{L}_o(G,x_0))(\forall n \geq K ) \nonumber \\
		&(\exists x_0'\in X_0)(\exists s'\in \mathcal{L}_o(G,x'_0), t\in \mathcal{L}(G,f(x'_0,s')) \\\nonumber
		&[P(s)=P(s')  \wedge  |t|=n  \wedge s't\notin \Omega]
		\end{align}
	\end{mydef}

	\begin{mydef}($K$-Step Trajectory Pattern Pre-Opacity)\label{def:ptraj}
	Given system $G$,  set of observable events $E_o$, a sequence pattern $\Omega$, and non-negative integer $K\in \mathbb{N}$, system $G$ is said to be $K$-step trajectory pattern pre-opaque (w.r.t.\ $E_o$ and $\Omega$) if 
	\begin{align}
	&(\forall x_0\in X_0,\forall s\in \mathcal{L}_o(G,x_0))(\forall n\ge K)  \nonumber \\
	&(\exists x_0'\in X_0,\exists s'\in \mathcal{L}_o(G,x'_0) , \exists t_1t_2\in \mathcal{L}(G,f(x_0',s'))\text{ s.t. } \nonumber\\
	&[P(s)=P(s')]  \wedge  [|t_1|=K]\wedge  [|t_1t_2|=n]  \wedge \nonumber\\
	&\quad[\forall w\!\in\! \overline{\{t_2\}}:s't_1w\notin		 \Omega ]\nonumber
	\end{align}
    \end{mydef}

	Intuitively, $K$-step trajectory pattern pre-opacity says that, for any observation, the intruder cannot predict $K$-step ahead that a secret sequence will be completed. 
	The definition of instant pattern pre-opacity is similar; the only difference is that it also requires to specify the specific instant of the completion. 
 	Clearly, pre-opacity is a special case of pattern pre-opacity as we can define all sequences reaching secret states as the sequence pattern. 
 	Hereafter, we will show that pattern pre-opacity can also be transformed to standard pre-opacity by refining the state-space and suitably defining secret states.

	\begin{myexm}
		Consider again the example shown in Figure~\ref{fig:G3}. 
		The secret sequence pattern can be described by the regular language
		\begin{align}
		   \Omega=
		   &((E\backslash\{\alpha_1\})^*\{\alpha_1\}(E\backslash\{\alpha_6\})^*\{\alpha_6\}(E\backslash\{\beta_{3}\})^*\{\beta_3\}\nonumber\\
		   &\cup  (E\backslash\{\alpha_1\})^*\{\alpha_1\} (E\backslash\{\beta_{2}\})^* \{\beta_2\})^*
		\end{align}
		Essentially, regular language $\Omega$ includes all strings that contain $\alpha_1\alpha_6\beta_3$ or $\alpha_1\beta_2$. 
		This language can be marked by DFA $G_{\Omega}$  shown in Figure~\ref{fig:marked}.
		Obviously, $G_3$ is 2-step instant pattern pre-opaque, since based on any observation, the intruder cannot know for sure the system will finish a sequence pattern $\alpha_1\beta_2$ or $\alpha_1\alpha_6\beta_3$ 2-step ahead.
		However, as we discussed early, it is not 1-step instant pattern pre-opaque;  this is because, once string $\alpha_1\alpha_6$ is observed, the monitor knows for sure that sequence $\alpha_1\alpha_6\beta_{3}$ will be completed in 1-step.		
		Also, we can check that $G_3$ is 2-step trajectory pattern pre-opaque but not 1-step trajectory pattern pre-opaque.
		
		Note that when string $\alpha_1\alpha_6$ is observed, we know that sequence $\alpha_1\beta_{2}$ has been finished one step ago.
		Although the monitor fails to detect pattern $\alpha_1\beta_{2}$ before its completion, it still can predict secret sequence pattern $\alpha_1\alpha_6\beta_{3}$.
	\end{myexm}

	\begin{remark}
		The concept of \emph{sequence pattern} was first proposed in  the literature for the purpose of fault diagnosis  \cite{jeron2006supervision} and fault prognosis  \cite{jeron2008pred}.   
		Specifically, a sequence pattern is used to model the set of behaviors considered as fault. 
		Our notion of sequence pattern is more general than that in the context of fault diagnosis/prognosis.  
		In particular, in the  context of fault diagnosis/prognosis, the sequence pattern is assumed to be \emph{stable} in the sense that any continuation of a sequence in the pattern is still in the pattern. This is motivated by the setting of permanent fault. However, our definition of sequence pattern does not necessarily be stable as the system can be secret/non-secret intermittently. 
		In other words, even if the intruder miss the predication of the first pattern, it may still be able to predict some future pattern, and in this case, the system is also not pre-opaque.
	\end{remark}

	\subsection{Verifications of Pattern Pre-Opacity}
	
	We show how to verify pattern pre-opacity in this part.
	To this end, we assume that the secret sequence pattern $\Omega$ is a regular language   and it is recognized by a DFA $G_\Omega=(X_{\Omega}, E, f_{\Omega}, x_{0,\Omega},X_{m,\Omega})$, i.e., $\mathcal{L}_m(G_\Omega)=\Omega$, where $x_{0,\Omega}$ is the unique initial state. 
	Without loss of generality, we assume that $G_\Omega$ is total, i.e., $\mathcal{L}(G_\Omega)= E^*$; otherwise, we can add a new unmarked ``dump" state and complete the transition function. 
	
	Then let $G=(X,E,f,X_0)$ be the system and $G_\Omega=(X_{\Omega}, E, f_{\Omega}, x_{0,\Omega},X_{m,\Omega})$ be the DFA recognizing the sequence pattern. 
	We define the product of $G$ and $G_\Omega$ as
	\[
	G_\times =(X',E',f',X'_0), 
	\]
	where $X'\subseteq X\times X_{\Omega}, E'=E, X'_0= X_0\times \{x_{0,\Omega}\}$ and $f': X'\times E\to X'$ is the transition function defined by 
	$f((x_1,x_2),\sigma)=(f(x_1,\sigma),f_{\Omega}(x_{2},\sigma))$, if $f(x_1,\sigma)$	and $f_{\Omega}(x_{2},\sigma)$ are defined, and undefined  otherwise. 
	Then we define
	\[ 
	X'_S=\{(q_1,q_2):q_2\in X_{m,\Omega}\}
	\]
	as the set of secret states in  $G_\times$. 
	Then the following result shows that pattern pre-opacity can be transformed to state-based pre-opacity. 
	
	\begin{mythm}\label{thm:pattern}
	System $G$ is  $K$-step instant (respectively, trajectory) pattern pre-opaque w.r.t. $\Omega$ if and only if
	$G\times G_{\Omega}$ is $K$-step instant  (respectively, trajectory)  pre-opaque w.r.t. $X'_S$.
    \end{mythm}
	\begin{proof}
		We only show the case of instant pre-opacity; the case of trajectory pre-opacity is similar. 
		
		($\Rightarrow$) Suppose that $G\times G_{\Omega}$ is not $K$-step instant pre-opaque, which implies that 
		\begin{align}
		&(\exists (x_0,x_{0,\Omega})\in X'_0)(\exists s\in \mathcal{L}_o(G\times  G_{\Omega},(x_0,x_{0,\Omega})))(\exists n_0 \geq K ) \nonumber \\
		&(\forall (x_0',x'_{0,\Omega})\in X'_0)\nonumber\\
		&(\forall s'\in \mathcal{L}_o(G\times  G_{\Omega},(x_0',x'_{0,\Omega})), s't\in \mathcal{L}(G\times  G_{\Omega},(x_0',x'_{0,\Omega}))) \nonumber\\
		&[P(s)=P(s')  \wedge  |t|=n_0]  \Rightarrow [f'((x_0',x'_{0,\Omega}),s't)\in X'_S]\nonumber
		\end{align} 
		Since $G_{\Omega}$ is complete, we have                                                                                                                                                                                                                                                                                                                                                                                                                                                                                                                                                                                                                                                                                                                                                                                                                                                                                                                                                                                                                                                                                                                                                                                                                                                                                                                                                                                                                                                                                                                                                                                                                                                                                                                                                                                                                                                                                                                                                                                                                                                                                                                                                                                                                                                                                                                                                                                                                                                                                                                                                                                                                                                                                                                                                                                                                                $\mathcal{L}(G)\subseteq \mathcal{L}(G_{\Omega})$, then we know that for any $x'_{0}\in X_0$, any $s'\in \mathcal{L}_o(G,x'_0), s't\in \mathcal{L}(G,x'_0) $ such that $P(s)=P(s')$ and $|t|=n_0\ge K$, we have that $f_{\Omega}(x'_{0},s't)\in X_{m,\Omega}$, i.e., $s't\in \mathcal{L}_m( G_{\Omega})=\Omega$.
		This implies that   $G$ is not $K$-step instant pattern pre-opaque
		
		($\Leftarrow$) 
		Assume that $G$ is not $K$-step instant pattern pre-opaque, i.e.,
		\begin{align}
		&(\exists x_0\in X_0)(\exists s\in \mathcal{L}_o(G,x_0))(\exists n_0 \geq K ) \nonumber \\
		&(\forall x_0'\in X_0)(\forall s'\in \mathcal{L}_o(G,x'_0), s't\in \mathcal{L}(G,x'_0)) \nonumber\\
		&[P(s)=P(s')  \wedge  |t|=n_0  \wedge s't\in \Omega]\nonumber
		\end{align}
		Since $\mathcal{L}(G\times  G_{\Omega})\subseteq \mathcal{L}(G)$, we know that for any $(x'_0,x_{0,\Omega})\in X'_0, s'\in \mathcal{L}_o(G\times  G_{\Omega},(x'_0,x_{0,\Omega}))$ and $s't\in \mathcal{L}(G\times  G_{\Omega},(x'_0,x_{0,\Omega}))$ such that $P(s')=P(s)$ and $|t|=n_0$, we always have  $f'((x'_0,x'_{0,\Omega}),s't)\in X'_S$ and $|t|=n_0\ge K$, which means that $G\times G_{\Omega}$ is not $K$-step instant pre-opaque.  
	\end{proof}

	\begin{myexm}
		Consider again system automaton $G_3$ and pattern automaton $G_{\Omega}$ in Figure~\ref{fig:exm2}. 
		To verify the $K$-step instant/trajectory pattern pre-opacity of $G_3$, we fist construct $G_3\times G_{\Omega}$, which is omitted here for the sake of brevity. 
		Then the set of secret states in $G_3\times G_{\Omega}$ is $X'_S=\{(6,F),(8,H)\}$.
		One can verify that $G_3\times G_{\Omega}$ is 2-step instant pre-opaque but not 1-step instant pre-opaque; 
		also, $G_3\times  G_{\Omega}$ is 2-step trajectory pre-opaque but not 1-step trajectory pre-opaque.
		Therefore, based on Theorem~\ref{thm:pattern}, for sequence pattern captured by $\Omega$, we know that $G_3$ is $2$-step instant pattern pre-opaque but not $1$-step instant pattern pre-opaque, and it is $2$-step trajectory pattern pre-opaque but not $1$-step trajectory pattern pre-opaque, which are consistent with our previous analysis. 
	\end{myexm}
	
	\section{Conclusion}\label{sec:6}
	In this paper, we proposed the notion of pre-opacity to verify the \emph{intention security} of a partially-observed DES. 
	Two notions of pre-opacity called $K$-step instant pre-opacity and $K$-step trajectory pre-opacity are proposed. 
	For each notion of pre-opacity, we provide a verifiable necessary and sufficient condition as well as an effective verification algorithm. We also generalize the notions of pre-opacity to the case where the secret behavior is captured by a sequence pattern. 
	Our work extends the theory of opacity to a new class where   secret is related to the intention of the system.  
	We believe there are many interesting future directions related to the concept of  pre-opacity. 
	One interesting  direction  is  to \emph{synthesize} a supervisor  to enforce pre-opacity when the verification result is negative. 
	Also, we would like to extend the notion of pre-opacity to the stochastic setting to quantitatively evaluate the information leakage.
	
	\bibliographystyle{plain}
	\bibliography{des}

\end{document}